\documentclass[runningheads]{llncs}
\usepackage[hidelinks]{hyperref}

\usepackage{appendix}
\usepackage[dvipsnames]{xcolor}
\usepackage{lipsum}
\usepackage{multirow}
\usepackage{graphicx}
\usepackage{amssymb,amsmath}
\usepackage{float}
\usepackage{tikz}
\usetikzlibrary{fadings}
\usetikzlibrary{patterns}
\usetikzlibrary{shadows.blur}
\usepackage{multicol}
\usepackage{enumitem}
\setlistdepth{9}
\renewlist{itemize}{itemize}{9}
\setlist[itemize]{label=$\cdot$}
\setlist[itemize,1]{label=\textbullet}
\setlist[itemize,2]{label=--}
\setlist[itemize,3]{label=*}

\newcommand{\R}{\mathbb{R}}

\newcommand{\accept}{\mathsf{accept}}
\newcommand{\reject}{\mathsf{reject}}
\newcommand{\evaluate}{\mathsf{evaluate}}
\newcommand{\validate}{\mathsf{validateTrue}}

\newcommand{\vrfy}{\mathsf{vrfySig}}

\newcommand{\hA}{\mathcal{A}}

\newcommand{\MT}{\mathsf{MT}}

\newcommand{\cR}{\mathsf{root}}
\newcommand{\cLedger}{\mathsf{ledger}}
\newcommand{\cHashes}{\mathsf{blockhashes}}
\newcommand{\cQ}{\mathsf{Q}}
\newcommand{\cH}{\mathsf{H}}

\newcommand{\cC}{\mathsf{C}}

\newcommand{\cD}{\mathsf{D}}
\newcommand{\cP}{\mathsf{P}}

\newcommand{\tx}{\mathsf{tx}}
\newcommand{\txid}{\mathsf{txid}}
\newcommand{\TX}{\mathsf{TX}}

\newcommand{\yuan}[1]{{\color{red}[Yuan: #1]}}

\newcommand{\ignore}[1]{}
\DeclareMathOperator*{\argmax}{arg\,max}

\def\bitcoin{%
	\leavevmode
	\vtop{\offinterlineskip 
		\setbox0=\hbox{B}%
		\setbox2=\hbox to\wd0{\hfil\hskip-.03em
			\vrule height .3ex width .15ex\hskip .08em
			\vrule height .3ex width .15ex\hfil}
		\vbox{\copy2\box0}\box2}}

\begin{document}
\title{Generic Superlight Client for Permissionless Blockchains}
%
%
\author{Yuan Lu \and
Qiang Tang \and
Guiling Wang}
\authorrunning{Y. Lu et al.}
%
\institute{New Jersey Institute of Technology, Newark NJ 07102, USA \\
\email{\{yl768, qiang, gwang\}@njit.edu}}
\maketitle              

\begin{abstract}
We conduct a systematic study on the  light-client protocol of permissionless blockchains, in the setting where full nodes and light clients  are rational. In the game-theoretic model, we design a superlight-client protocol to enable a light client to employ some relaying full nodes (e.g., two or one) to read the blockchain. The protocol is ``generic'', i.e., it can be deployed disregarding underlying consensuses, and it is also ``superlight'', i.e.,  the computational cost of the light  client to predicate the (non)existence of a transaction in the blockchain becomes a small constant. Since our   protocol resolves a fundamental challenge of broadening the usage of blockchain technology, it captures a wide variety of important use-cases such as multi-chain wallets, DApp browsers and more.
\end{abstract}


\section{Introduction}


The blockchain \cite{Nak08,buterin2014next} can be considered as an abstracted global ledger  \cite{GKL15,KMS16,badertscher2017bitcoin},
which can be read and written by the users of higher level applications \cite{buterin2014next,Filecoin,DWA17,fairswap}. 
However, the basic abstraction of \emph{reading the   ledger}\footnote{\label{note:lightwriting}Writing in  the blockchain is trivial, as one can   gossip with some full nodes to diffuse its messages to the entire blockchain network (a.k.a., network diffuse functionality \cite{GKL15,BDO12}). Then the blockchain's liveness ensures the inclusion of the messages \cite{GKL15}.} 
implicitly requires the user       maintain a personal  {\em full node}  \cite{btccore,ethgo} to execute the consensus protocol and maintain a local blockchain replica.
Nevertheless, with the rapid popularity of blockchain, an increasing number of blockchain  users become  merely caring about  the high-level applications such as cryptocurrency, instead of maintaining any personal full node  \cite{GLK18};
let alone, many of them  are resource-starved, say
browser extensions and smartphones  \cite{TD17,CryptoKitties,Sarah16,steemit},  that have too limited resources to stay on-line to stick with the consensus protocol.

Thus an urgent demand of blockchain's {\em light clients}, a.k.a., superlight clients or lightweight clients \cite{Flyclient,powsidechain,KMZ17,VitalikPoS}, rises up.
%
Consider a quintessential scenario: 
Alice is the  cashier of a pizza store;
a customer Bob tells her   \bitcoin1,000 has been paid for some pizzas, via a bitcoin transaction   with $\txid$ \href{https://www.blockchain.com/btc/tx/a1075db55d416d3ca199f55b6084e2115b9345e16c5cf302fc80e9d5fbf5d48d}{$\mathtt{0xa1075d...}$}, and claims the transaction is already in the blockchain;
then, Alice needs to check that, by activating a lite wallet app installed in her mobile phone.
That is to say, 
Alice, and many typical blockchain users, need    a superlight client, which can   {\em stay  off-line} to opt out of the consensus,
 and  can  still wake  up   any time  to ``read'' the records in the blockchain  with high security and low computational cost.
%

\subsection{Insufficiencies of existing practices}

The fundamental challenge of designing the superlight client   stems from a fact:
  the records in the blockchain are    
``authenticated'' by the latest chain agreed across the blockchain network (a.k.a., the main-chain) \cite{pass2017rethinking,Nak08,snowwhite,algorand,ouroboros}.
So without running the consensus to obtain a trusted replica of the main-chain  at hand, 
the   client has to rely on some other   {\em full nodes} to forward the records in the blockchain.
That   said, the light-client protocol therefore must deal with those probably distrustful full nodes that can forward fake blockchain readings.

\smallskip
\noindent{\bf Some ad-hoc attempts}. 
 A few   proposals attempt to  prevent the   client   being cheated, by relying on    heavyweight   assumptions.
%
%
For example, a few proposals \cite{electrum,snowwhite} assume a diverse list of   known full nodes to serve as relays to forward blockchain readings. 
The list is supposed to consist of some ``mining'' pools and  a few so-called blockchain ``explorers'',
so the client can count on the honest-majority of these known relays to read the chain.
But for many real-world permissionless blockchains, this  assumption is too heavy to hold with high-confidence. 
Say  Cardano \cite{Cardano}, a top-10 blockchain by market capitalization, has quite few ``explorer'' websites on the run;
even worse, the naive idea of recruiting ``mining'' pools  as relays is  more elusive, considering  
most of them would not participate without moderate incentives  \cite{GLK18}.
So it is unclear how to identify an honest-majority set of  known relays for each  permissionless blockchain in the wild.
As such, these ad-hoc solutions become unreliable, considering their heavyweight assumptions  are seemingly elusive in practice.

\smallskip
\noindent{\bf Cryptographic approaches}.
To design the light-client protocol against malicious relay full nodes,
a few cryptographic approaches are proposed \cite{Flyclient,KMZ17,Nak08,EPBC}. 

\underline{\smash{\em Straight use of SPV is problematic}}. The most straightforward way to instantiating the idea is to let the    client  keep   track of  the suffix of the main-chain, and then check the existence of transactions by verifying SPV proofs \cite{Nak08}, but the naive approach causes at least one major  issue:
the client has to frequently be  on-line to track the growth of main-chain. Otherwise, when  the client wakes up from a deep sleep, it   needs to at least verify the block headers of the main-chain linearly. Such the bootstrapping can be costly, considering the main-chain is ever-growing, say the   headers of Ethereum is growing at a pace of $\sim1$ GB per year.
As a result, in many critical use-cases such as web browsers and/or mobile phones, the    idea of   straightly using SPV proofs  becomes rather unrealistic.
%
%

%
%
\underline{\smash{\em PoW-specific results}}. 
For PoW chains, some existing superlight clients such as FlyClient and NiPoPoW \cite{Flyclient,KMZ17} circumvent the problem of   SPV proofs.
These ideas notice the   main-chain is essentially ``authenticated'' by its {\em  few suffix blocks},
and then develop some PoW-specific techniques to allow 
the {\em suffix blocks} be proven to  a client at only sublinear cost.
%
%
But they come with one major limit, namely, need to verify   PoWs to discover the correct suffix, and therefore cannot fit the promising class of proof-of-stake (PoS) consensuses \cite{ouroboros,snowwhite,algorand}.
%

%

%
%
%
%
%
%
%

\underline{\smash{\em Superlight client for PoS still unclear}}. For  PoS chains, it is yet unclear how to realize an actual superlight client that can go off-line to completely opt out  of consensus. 
The major issue is lacking an efficient way to proving the suffix of   PoS  chains to an off-line client,
as the validity of the suffix blocks  relies on the signatures of stakeholders, 
whose validities further depend on the recent stake distributions,
which   further are authenticated by the blockchain itself \cite{algorand,snowwhite,ouroboros,VitalikPoS}.
%
Some recent efforts \cite{GKZ19,Coda,Cosmos} allow the \emph{always-online} clients
to use minimal space to track  the suffix   of  PoS chains, without maintaining the stake distributions. 
But   for the major challenge of enabling the clients to go off-line, there only exist few fast-bootstrapping proposals for {\em full nodes} \cite{vault,ouroboros-genesis}, which still require the client to download and verify a linear portion of the  PoS chain to respawn, thus  incurring substantial cost \cite{algorand,ouroboros-praos}.
%

\smallskip
\noindent{\bf Demands of ``consensus-oblivious'' light client.} 
Most existing light-client protocols are highly specialized for concrete consensuses (e.g., PoW).
This not only prevents us   adapting them to instantiate actual superlight clients for the PoS chains,
but also hinders their easy deployment and user experience in many important use-cases in reality.
A typical scenario is a multi-chain wallet, which is expected to support various cryptocurrencies atop different chains, each of which is even running a distinct consensus. 
Bearing that   existing light-client protocols are highly specialized \cite{GKZ19,Coda,Cosmos,Flyclient,KMZ17}, 
the multi-chain wallet needs to instantiate different   protocols for distinct consensuses.
That said, without a generic solution, the multi-chain wallet ends up to contain many independent ``sub-wallets''.
That wallet not only is  burdensome for the  users, but also challenges  the developers to correctly implement all sub-wallets. 
In contrast, if there is a generic protocol fitting all, one can simply tune some parameters at best.


\smallskip
\noindent{\bf Explore a generic solution  in a different setting.} 
The cryptographic setting seems to be   an inherent obstacle-ridden path to the  generic light-client protocol.
Recall it  usually needs to prove the main-chain's suffix  validated by   the consensus rules  \cite{GKZ19,Cosmos,Flyclient,KMZ17}. Even if one  puts forth a ``consensus-oblivious'' solution in the cryptographic setting,
it  likely has to convert all ``proofs'' for suffix into a generic statement of verifiable computation (VC) \cite{Coda}, 
which is unclear how to be realized practically, considering VC itself is not  fully practical yet for complicated   statements 
(see Appendix \ref{related} for a  thorough review on the pertinent topics).
%

To meet the urgent demand of the generic light-client protocol,
we explicitly deviate from     the  cryptographic setting, and focus on the light-client problem due to the game-theoretic approach, in light of many successful studies such as 
rational multi-party computation \cite{ADGH06,groce2012fair,HT04,GK06,IML05,cryptoeprint:2011:396,kol2008games,fuchsbauer2010efficient}  and rational verifiable computation  \cite{DWA17,truebit,PKC14,Kupcu17}.
In the rational setting, we can hope that a consensus-independent incentive mechanism exists to assist a simple  light-client protocol,
so all rational protocol participants (i.e., the relay full nodes and the light client)   follow the protocol for selfishness. As a result, the client can efficiently retrieve the correct information about the blockchain, and the full nodes are fairly paid.

\medskip
Following that, this paper would present:
{\em  a systematic treatment to the  light-client problem of   permissionless blockchains in the game-theoretic setting. }


\subsection{Our results}



By reconsidering the light-client problem through the powerful lens of  the {\em  game-theoretic setting}, we
design a  superlight protocol to enable a light client to recruit several relay full nodes (e.g., one\footnote{Note that the case where only one relay is recruited models a pessimistic scenario that all recruited full nodes are colluding to form a single coalition.} or two) to securely evaluate a general class of predicates about the blockchain.

\smallskip
\noindent{\bf Contributions}. To summarize, our technical contributions are three-fold:
\begin{itemize}
	\item Our light-client protocol can be bootstrapped in the rational setting, efficiently and generically.
	First, the protocol is superlight, in the sense that the client   can go off-line and wake up any time to {\em evaluate} a general class of chain predicates	at a tiny constant computationally cost;  as long as the truthness or falseness of these chain predicates is reducible to few transactions' inclusion in the blockchain.
	Moreover, our generic protocol gets rid of the dependency on consensuses and can be deployed in nearly any permissionless blockchain (e.g.,   Turing-complete blockchains \cite{buterin2014next,Woo14})
	without even velvet forks \cite{zamyatin2018wild}, thus   supporting   the promising PoS type of consensuses.

	\item Along the way, we conduct a systematic study to understand whether, or to what extent, our protocol for the superlight client is secure, in the rational setting (where is no always available  trusted third-parties).
	We make non-trivial analyses of the incomplete-information extensive game induced by our light-client protocol, 
	and conduct a comprehensive study  to understand how to finely tune the incentives to achieve security in different scenarios, from the standard setting of non-cooperative full nodes to the pessimistic setting of colluding full nodes.

	\item Our protocol   allows the rational client to {\em evaluate} (non)existence of a given transaction.
	As such, a rational light client can be convinced by rational full nodes that a given transaction is   {\em not} in any block of the entire chain,
	which provides a simple   way to performing non-existence ``proof''.
	In contrast, to our knowledge, relevant studies in the cryptographic setting either give up non-existence proof \cite{Nak08,powsidechain}, 
	or require to heavily modify the data structure of the current blockchains \cite{IOPs,miller2014authenticated}.
\end{itemize}


\smallskip
\noindent{\bf Solution in a nutshell}.
Assuming the light client and relay nodes are rational, we leverage the smart contract to facilitate a simple and useful incentive mechanism,
such that being honest becomes their best choice, namely, for the highest utilities, 
(i) the relay nodes must forward the blockchain's records correctly, 
and (ii) the  client must pay the relays honestly. 
From high-level, the ideas are:
\begin{itemize}
	\item   \underline{\smash{\em Setup}}.
	The light client and relay node(s)   place their initial deposits in an ``arbiter'' smart contract, 
	such that a carefully designed incentive mechanism later can leverage these deposits to 
	facilitate rewards/punishments to deter deviations from the light-client protocol.
	
	
	\item   \underline{\smash{\em Repeated queries}}.
	After  setup, 
	the   client can 
	repeatedly query the relays to forward blockchain readings (up to $k$ times). Each query  proceeds as:
	\begin{enumerate} 
		\item {\em Request}. 
		The client firstly specifies the details of the predicate to query in the arbiter contract, 
		which can be done since writing in the contract is trivial for the network diffuse functionality.$^1$
		
		
		\item {\em Response}. Once the relays see the  specifications of the chain predicate   in the  arbiter contract, they are incentivized to evaluate the predicate and forward  the ground truth   to the   client off-chain.
		
		\item {\em Feedback}. 
		Then 
		the   client decides an output, according to what it receives from the relays.
		Besides, the client shall report what it receives to the arbiter contract; otherwise,   gets a fine.
		\item {\em Payout}. Finally, the    contract  verifies whether  the relays are honest, according to the feedback from the client, and then facilitates an incentive mechanism to reward (or punish)  the relays.
	\end{enumerate}
	Without a proper incentive mechanism, the above simple protocol is insecure to any extent as it is. So carefully designed incentives are added in the arbiter contract to ensure ``following the protocol'' to be  a desired   equilibrium, 
	namely, the rational   relays and   client, would not deviate.
\end{itemize}


\vspace{-1cm}
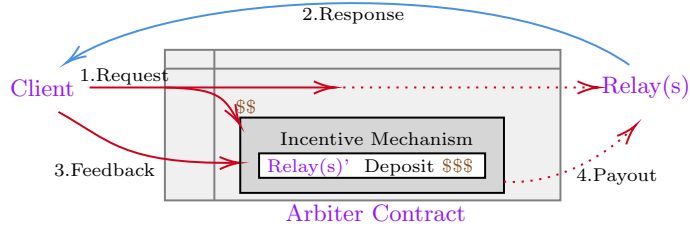
\begin{figure}[!htb]
	\centering
	
		\tikzset{every picture/.style={line width=0.75pt}} 
		
		\begin{tikzpicture}[x=0.75pt,y=0.75pt,yscale=-0.95,xscale=0.95]
		
		\draw  [color={rgb, 255:red, 128; green, 128; blue, 128 }  ,draw opacity=1 ][fill={rgb, 255:red, 155; green, 155; blue, 155 }  ,fill opacity=0.15 ] (210,120) -- (420,120) -- (420,200) -- (210,200) -- cycle ; \draw  [color={rgb, 255:red, 128; green, 128; blue, 128 }  ,draw opacity=1 ] (236.25,120) -- (236.25,200) ; \draw  [color={rgb, 255:red, 128; green, 128; blue, 128 }  ,draw opacity=1 ] (210,130) -- (420,130) ;
		\draw [color={rgb, 255:red, 208; green, 2; blue, 27 }  ,draw opacity=1 ]   (170,140) -- (298,140) ;
		\draw [shift={(300,140)}, rotate = 180] [color={rgb, 255:red, 208; green, 2; blue, 27 }  ,draw opacity=1 ][line width=0.75]    (10.93,-3.29) .. controls (6.95,-1.4) and (3.31,-0.3) .. (0,0) .. controls (3.31,0.3) and (6.95,1.4) .. (10.93,3.29)   ;
		
		\draw [color={rgb, 255:red, 208; green, 2; blue, 27 }  ,draw opacity=1 ] [dash pattern={on 0.84pt off 2.51pt}]  (300,140) -- (438,140) ;
		\draw [shift={(440,140)}, rotate = 180] [color={rgb, 255:red, 208; green, 2; blue, 27 }  ,draw opacity=1 ][line width=0.75]    (10.93,-3.29) .. controls (6.95,-1.4) and (3.31,-0.3) .. (0,0) .. controls (3.31,0.3) and (6.95,1.4) .. (10.93,3.29)   ;
		
		\draw [color={rgb, 255:red, 74; green, 144; blue, 226 }  ,draw opacity=1 ]   (456.5,128) .. controls (388.84,80.24) and (213.27,88.91) .. (158.32,126.43) ;
		\draw [shift={(157.5,127)}, rotate = 324.87] [color={rgb, 255:red, 74; green, 144; blue, 226 }  ,draw opacity=1 ][line width=0.75]    (10.93,-3.29) .. controls (6.95,-1.4) and (3.31,-0.3) .. (0,0) .. controls (3.31,0.3) and (6.95,1.4) .. (10.93,3.29)   ;
		
		\draw [color={rgb, 255:red, 208; green, 2; blue, 27 }  ,draw opacity=1 ]   (153.5,153) .. controls (180.39,167.85) and (187.03,180.74) .. (248.13,180.03) ;
		\draw [shift={(250,180)}, rotate = 539.0899999999999] [color={rgb, 255:red, 208; green, 2; blue, 27 }  ,draw opacity=1 ][line width=0.75]    (10.93,-3.29) .. controls (6.95,-1.4) and (3.31,-0.3) .. (0,0) .. controls (3.31,0.3) and (6.95,1.4) .. (10.93,3.29)   ;
		
		\draw  [color={rgb, 255:red, 0; green, 0; blue, 0 }  ,draw opacity=1 ][fill={rgb, 255:red, 155; green, 155; blue, 155 }  ,fill opacity=0.3 ] (250,156) -- (390,156) -- (390,196) -- (250,196) -- cycle ;
		\draw [color={rgb, 255:red, 208; green, 2; blue, 27 }  ,draw opacity=1 ] [dash pattern={on 0.84pt off 2.51pt}]  (390,190) .. controls (413.56,191.95) and (440.77,179.64) .. (458.64,161.41) ;
		\draw [shift={(460,160)}, rotate = 493.19] [color={rgb, 255:red, 208; green, 2; blue, 27 }  ,draw opacity=1 ][line width=0.75]    (10.93,-3.29) .. controls (6.95,-1.4) and (3.31,-0.3) .. (0,0) .. controls (3.31,0.3) and (6.95,1.4) .. (10.93,3.29)   ;
		
		\draw [color={rgb, 255:red, 208; green, 2; blue, 27 }  ,draw opacity=1 ]   (210,140) .. controls (238.61,140.68) and (241.34,144.7) .. (249.13,158.46) ;
		\draw [shift={(250,160)}, rotate = 240.46] [color={rgb, 255:red, 208; green, 2; blue, 27 }  ,draw opacity=1 ][line width=0.75]    (10.93,-3.29) .. controls (6.95,-1.4) and (3.31,-0.3) .. (0,0) .. controls (3.31,0.3) and (6.95,1.4) .. (10.93,3.29)   ;
		
		\draw  [color={rgb, 255:red, 0; green, 0; blue, 0 }  ,draw opacity=1 ][fill={rgb, 255:red, 255; green, 255; blue, 255 }  ,fill opacity=1 ] (260,175) -- (380,175) -- (380,188) -- (260,188) -- cycle ;
		
		\draw (322.5,172) node [anchor=south] [inner sep=0.75pt]  [font=\scriptsize] [align=left] {Incentive Mechanism};
		\draw (322,206.5) node  [font=\small,color={rgb, 255:red, 144; green, 19; blue, 254 }  ,opacity=1 ] [align=left] {Arbiter Contract};
		\draw (145,140) node  [font=\small,color={rgb, 255:red, 144; green, 19; blue, 254 }  ,opacity=1 ] [align=left] {Client};
		\draw (465.5,140) node  [font=\small,color={rgb, 255:red, 144; green, 19; blue, 254 }  ,opacity=1 ] [align=left] {Relay(s)};
		\draw (188.5,134) node  [font=\scriptsize] [align=left] {1.Request};
		\draw (178,183.5) node  [font=\scriptsize] [align=left] {3.Feedback};
		\draw (310,102) node  [font=\scriptsize] [align=left] {2.Response};
		\draw (450,187.5) node  [font=\scriptsize] [align=left] {4.Payout};
		\draw (253,150) node  [font=\scriptsize,color={rgb, 255:red, 139; green, 87; blue, 42 }  ,opacity=1 ]  {${\displaystyle \$\$}$};
		\draw (321,182.5) node  [font=\scriptsize] [align=left] {\textcolor[rgb]{0.56,0.07,1}{Relay(s)' }  Deposit $\displaystyle \textcolor[rgb]{0.55,0.34,0.16}{\$\$\$}$ };
	
		\end{tikzpicture}
	\caption{Superlight client in the rational model (high-level).}
\label{fig:intuition}
\end{figure}
\vspace{-0.5cm}

%
\smallskip
\noindent{\bf Challenges \& techniques}.
Though   instantiating the  above    idea seems simple, it on the contrary  is challenged by the limit  of the ``handicapped''  arbiter contract.
In particular,  the arbiter contract in most blockchains (e.g., Ethereum) cannot directly verify the non-existence of  transactions,
though it is easy to verify any transaction's existence if being given the corresponding inclusion proof \cite{BCD14,powsidechain}.
Thus, for a chain predicate whose trueness (resp. falseness) is reducible to the  existence (resp. nonexistence) of some transactions,
the arbiter contract can at most verify either its trueness or its falseness, but not both sides.

	%
	%
	%
	In other words, considering the chain predicate as a binary $True$ or $False$ question, there exists a proof verifiable by the arbiter contract to attest   it is $True$ (or it is $False$), but not both.
	%
	%
	This enables the relays to adopt a malicious strategy:  ``always  forward  unverifiable bogus disregarding the actual ground truth'', because doing so would not be caught by the   contract, and thus the relays are still paid. 
	%
	%
	%
	The   challenge therefore becomes   how to  design an  incentive mechanism  to    deter the relays from   flooding unverifiable bogus claims, given only the ``handicapped'' verifiability of the   contract.
	%

	

%
To circumvent the  limit of the arbiter contract, we  squeeze the most of   its ``handicapped'' verifiability  to finely tune the incentive mechanism, such that ``flooding fake unverifiable claims'' become irrational.
Following that,  any deviations from the protocol are further deterred,
%
%
from the standard setting of non-cooperative  relays to the extremely hostile case of colluding relays:  
	
	
	
	\begin{itemize}
		
		\item  If  two {\em   non-cooperative relays} (e.g., two competing mining pools in practice) can be identified and recruited,
		we   leverage the natural tension between these   {\em two} selfish relays to ``audit'' each other. As such, fooling the   client is deterred,
		because  a selfish relay is incentivized to report the other's (unverifiable) bogus claim, by producing a proof attesting the opposite of the fake claim.
		%
		%
		%

		\item In the extremely adversarial scenario where  any two  recruited relays can {\em     form a coalition},
		the setting becomes rather pessimistic,
		as the client is essentially   requesting an unknown knowledge from a {\em  single} party.
		Nevertheless, the incentive   can still be slightly tuned    to function as follows:
		\begin{enumerate}
			\item The first tuning does {\em not} rely on any extra assumption.
			%
			The  adjustment is to let
			the arbiter contract assign a higher payoff to a proved claim while make a lower payoff to an unprovable claim. 
			So the best strategy of the {\em only} relay is to  forward the actual ground truth, 
			as long as the malicious benefit attained by fooling the client is smaller than the maximal reward promised by the   client.
			
			Though this result has   limited applicabilities, for example, cannot handle valuable queries,
			it  still captures a variety of meaningful real-world use-cases, in particular,
			many  DApp browsers, where the relay  is not rather interested in cheating the client.

			\item The second adjustment relies on    another moderate rationality assumption, that is: at least one selfish {\em public} full node (in the entire blockchain network) can keep on monitoring the internal states of the arbiter contract at a tiny cost and {\em will not cooperate} with the  recruited relay. 

			Thus whenever the recruited relay   forwards an {\em unprovable} bogus claim to the client,
			our design incentivizes the {\em selfish} public full node to ``audit'' the relay by proving the opposite side is the actual ground truth,
			which deters the recruited  relay  from   flooding unprovable bogus.
			

		\end{enumerate}

		%

	\end{itemize}


\noindent{\bf Application scenarios}.
Our protocol    supports a wide variety of applications, as it solves a fundamental issue preventing low-capacity users using blockchain:
\begin{itemize}
	\item \emph{Decentralized application browser}. 
	The DApp browser is a natural application scenario. For example, a lightweight browser for CryptoKitties \cite{CryptoKitties} can get rid of a trusted Web server.
	When surfing the DApp via a distrustful Web server, the users   need to verify whether the content rendered by the    server is correct,
	which can be done through our light-client protocol efficiently.
	\item 
	\emph{Mobile wallet for multiple cryptocurrencies}. 
	Our protocol can be leveraged to implement a super-light mobile wallet to verify the (non)existence of cryptocurrency transactions. In particular, it can keep track of multiple coins atop   different blockchains running over diverse types of consensuses. 
\end{itemize}

		
\section{Warm-up: game-theoretic security}\label{warmup}
%
Usually,
the game-theoretic analysis of an interactive protocol   
  starts by defining an    extensive game
\cite{HPS16,DWA17,park2018spacemint} 
to model the strategies (i.e., probabilistic interactive Turing machines) of each party in the protocol.
Then a  utility function   would assign every party a certain payoff, for each possible execution   induced by the strategies of all parties.
So the security of the protocol can be argued by the properties of the   game,
for example, its Nash equilibrium  \cite{HT04} or other stronger equilibrium notions  \cite{halpern2019sequential,HPS16,DWA17,park2018spacemint,kol2008games}.
Here we introduce a simple interactive ``protocol'' to exemplify the idea of  conducting analysis in game-theoretic model, while deferring  the   extensive preliminary definitions to Appendix \ref{append:seq}.

\subsection{An interactive protocol as an extensive-form game}

Consider   an   oversimplified    ``light-client protocol'':
Alice is a cashier of a pizza store; her client   asks a  full node (i.e. relay) to check  a transaction's (non)existence, 
and simply  terminates to output what is forwarded by the relay.

\smallskip
\noindent{\bf Strategy, action, history, and information set}.
Let the {\em oversimplified} ``protocol'' proceed in synchronous round. 
In each round,  the parties will execute its {\em strategy}, i.e., a probabilistic polynomial-time ITM in our context, to produce and feed a string to the protocol, a.k.a., take an {\em action}. 
During the course of the protocol, a sequence of actions would be made, and we say it is a {\em history} by convention of the game theory literature;
moreover, when a party  acts, it might have learned some (incomplete) information from earlier actions taken by other parties,
so the notion of {\em information sets} are introduced to characterize what has and has not been learned by each party
(see Appendix \ref{append:seq} for the deferred formal definitions).
%
Concretely speaking, the {\em oversimplified} ``light-client protocol'' can be described by the    extensive-form game as  shown  in Fig \ref{fig:warmup}:
\begin{enumerate}
	\item 	{\em Round 1 (chance acts)}. A definitional virtual    party called $chance$ 
	 sets the  ground truth,
	namely, it determines   $True$ or $False$ to represent whether the transaction exists (denoted by $a$ or $a'$ respectively). To capture the  uncertainty of the ground truth, the chance   acts arbitrarily.
	\item 	{\em Round 2 (relay  acts)}. Then, the relay is activated to forward $True$ or $False$ to the light client, which states whether the transaction exists or not. 
	Note   the strategy chosen by the relay is an ITM that can produce arbitrary strings in this round, 
	we need to map the  strings into the admissible actions, namely,
	$t$, $f$ and $x$.  For definiteness, we let the string of ground truth be interpreted as the action $t$,
	the string of the opposite of ground truth be interpreted as the action $f$,
	and all other strings (including    abort) be interpreted as $x$.
	
	
	\item 	{\em Round 3 (client  acts)}. Finally, the client outputs $True$ (denoted by $A$) or $False$  (denoted by $A'$) to represent whether the transaction exists or not, according to the (incomplete) information acquired from the protocol. Note the client knows how the relay acts, but cannot directly infer the action of  $chance$. So it faces three distinct information sets $I_1$, $I_2$ and $I_3$, which respectively represent the client  receives $True$, $False$ and others  in Round 2. The client cannot distinguish the histories inside each information set.
\end{enumerate}

 \vspace{-0.75cm}
\begin{figure} 
	\centering
	\includegraphics[width=8cm]{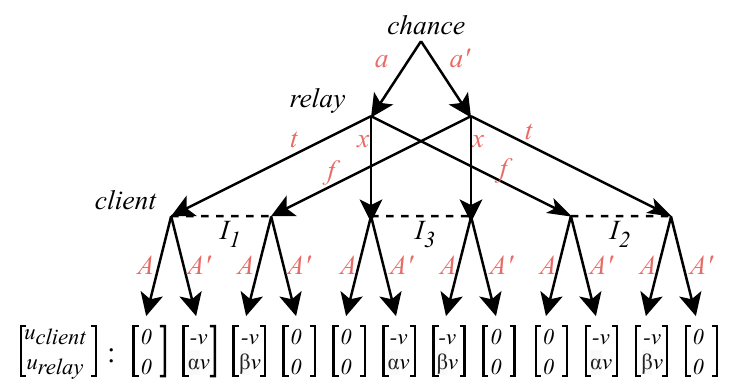}
	
	\caption{The extensive game of an oversimplified light-client ``protocol''.
		The utility function is an example to clarify the {\em insecurity} of such a trivial idea.
	}
	\label{fig:warmup}
	\vspace{-0.5cm}
\end{figure}

\smallskip
\noindent{\bf Utility function}.
After the protocol terminates, its game  reaches a so-called {\em terminal history}.
A well-defined {\em utility function}   specifies the economic outcome of each party, for each   terminal history    induced by the    extensive  game.

In practice, the utility function is determined by  some economic factors of the parties and   the protocol itself \cite{HT04,DWA17}.
%
For example, 
the rationale behind the utility function in Fig  \ref{fig:warmup} can be understood as:
(i) the relay is motivated to fool the client to believe the nonexistence of an existing transaction, because this literally ``censors'' Alice to harm her business by a loss of $\$v$, which also brings a malicious benefit $\$ \alpha\cdot v$ to the relay;
(ii) the relay also prefers to fool the client to believe the existence of a non-existing transaction, 
so the relay   gets   free pizzas valued by $\$\beta\cdot v$, which causes Alice lose $\$v$ (i.e., the amount supposed to be transacted to purchase pizzas),
(iii) after all, the oversimplified  protocol itself does not facilitate any punishment/reward, so will not affect the utility function.


\subsection{Security via equilibrium}

Putting the game structure and the utility function together,
we can argue the (in)security due to the equilibria in the    game.
In particular, 
we can  adopt the strong notion of {\em sequential equilibrium} for extensive games \cite{halpern2019sequential,HPS16,DWA17,park2018spacemint}
to demonstrate that the rational parties would not deviate,  at each stage during the execution of the protocol.
As a negative lesson, 
the oversimplified  ``light-client game'' in Fig  2 is {\em insecure} in the game-theoretic setting, as the relay can unilaterally deviate  to fool the client for higher utility. In contrast, if the protocol is {\em secure} in   game-theoretic settings,  its game shall realize desired equilibrium, such that rational parties would not diverge from the protocol for   highest utilities.


\section{Preliminaries}\label{preliminary}

\smallskip
\noindent{\bf Blockchain addressing}.
A blockchain (e.g., denoted by $\cC$) is a chain of block (headers). 
Each block  commits a list of payload (e.g., transactions).
Notation-wise, we use   Python bracket  $\cC[t]$ to address the block (header) at the height $t$ of the chain $\cC$. 
For example, $\cC[0]$ represents the genesis block, and $\cC[0:N]$ represents a chain consisting of $N$ blocks (where $\cC[0]$ is known as genesis). 
W.l.o.g., a block $\cC[t]$ is defined   as a tuple of $(h_{t-1}, nonce, \cR)$, where $h_{t-1}$ is the hash of the block $\cC[t-1]$, $nonce$ is the valid   PoX (e.g., the correct preimage in PoW, and the valid signatures in PoS), and $\cR$ is   Merkle tree root of  payload. 
Through the paper,   $\cC[t].\cR$  denotes   Merkle $\cR$ of block $\cC[t]$.

\smallskip
\noindent{\bf Payload \& Merkle tree}.
Let  $\TX_t := \langle \tx_1, \tx_2, \cdots, \tx_n \rangle$ denote a sequence of transactions that is the payload of the block $\cC[t]$.
Recall   $\TX_t$ is included by the block $\cC[t]$ through Merkle tree   \cite{Nak08,Woo14},
%
%
%
%
which is an authenticated data structure scheme   of  three algorithms $(\mathsf{BuildMT}, \mathsf{GenMTP}, \mathsf{VrfyMTP})$.
$\mathsf{BuildMT}$ inputs    $\TX_t= \langle \tx_1, \cdots, \tx_n \rangle$ and outputs a Merkle tree $\MT$ with $\cR$. 
$\mathsf{GenMTP}$ takes the   tree $\MT$ (built for $\TX_t$) and a transaction $\tx \in \TX_t$ as input, and outputs a proof $\pi_j$ for the inclusion of $\tx$ in $\TX_t$ at the position $j$.
$\mathsf{VrfyMTP}$ inputs   $\pi_j$,  $\cR$ and  $\tx$  and outputs either 1 or 0. 
%
%
%
%
%
%
%
%
The Merkle tree scheme satisfies:
(i) \emph{Correctness}. $\Pr[\mathsf{VrfyMTP}(\MT.\cR,\tx,\pi_i)=1 \mid \pi_i\leftarrow\mathsf{GenMTP}(\MT, \tx)$, $\MT\leftarrow\mathsf{BuildMT}(\TX)]=1$;
(ii) \emph{Security}. for  $\forall$ P.P.T.   $\hA$,   $\Pr[\mathsf{VrfyMTP}(\MT.\cR,\tx,\pi_i)=1 \wedge \tx \neq \TX[i] \mid \pi_i\leftarrow \hA(1^\lambda, \MT, \tx)$, $\MT\leftarrow\mathsf{BuildMT}(\TX)$ $] \leq negl(\lambda)$.
%
%
%
The detailed construction of the Merkle tree scheme is deferred to Appendix \ref{append:merkle}.


\smallskip
\noindent{\bf Smart contract}.
%
%
Essentially, a smart  contract \cite{buterin2014next,Woo14} can be abstracted as an ideal functionality with a global ledger subroutine, so it can faithfully instruct the ledger to freeze ``coins'' as deposits and then correctly facilitate conditional payments \cite{KMS16,KZZ16,badertscher2017bitcoin}. 
This paper explicitly  adopts  the widely-used   notations invented by Kosba \emph{et al.} \cite{KMS16} to  describe  the smart contract, for example: 
\begin{itemize}
	\item The   contract can access the global time $T$, which can be seen as an equivalent notion of the height of the latest blockchain.
	\item The contract can access  a global dictionary  $\cLedger$  for   conditional payments.
	\item We   slightly enhance their notations to allow the contract to access a global dictionary $\cHashes$. Each item $\cHashes[t]$ is the hash of the block $\cC[t]$.\footnote{ Remark that the above modeling requires the block hashes can be read by smart contracts from the blockchain's internal states (e.g. available global variables) \cite{solidity}. 
	In Ethereum, this  currently can be realized via the proposal of Andrew Miller \cite{blockhash}  and will be   incorporated   due to the already-planned Ethereum enhancement EIP-210 \cite{EIP210}. 
	}

	\item The contract would not send its internal states to the light client, which captures the client opts out of   consensus. However, the   client can send messages to the contract, due to the well abstracted network diffusion functionality.
\end{itemize}

In addition, we emphasize  that the blockchain can be seen as a global ledger functionality \cite{badertscher2017bitcoin,KZZ16} that allows all full nodes to maintain their local blockchain replicas consistent to the global dictionary $\cHashes$ (within a clock period).


\section{Problem Formulation}\label{problem}

The light-client protocol   involves a light client, some relay full nodes (e.g. one or two), and an ideal functionality (i.e. ``arbiter'' contract). The light client relies on the relays to ``read'' the  chain, and the  relays expect to receive correct payments.

\subsection{Formalizing readings from the blockchain}
\label{sec:predicate}

The basic functionality of our light-client protocol is to allow the resource-starved clients to evaluate the falseness or trueness about some statements over the blockchain  \cite{KMZ17}.  This aim is subtly broader than \cite{powsidechain}, whose goal is restricted to prevent the client from deciding trueness when the statement is actually false.

\smallskip
\noindent{\bf Chain predicate}.
The paper focuses on a general class of   chain  predicates whose trueness (or falseness) can be induced by up to $l$ transactions' inclusions in the chain, such as ``whether the transaction with identifier $\txid$ is in the blockchain $\cC[0:N]$ or not''. 
%
Formally,   we focus on the chain  predicate in the form of:
		$$\cP^\ell(\cC[0:N])=
		\begin{cases}
			False \text{, otherwise} \\
			True \text{, $\exists \cC^\prime \subset \cC[0:N]$  s.t. $ \cD^\ell(\cC^\prime)=True$}\\
		\end{cases}$$
	
	\noindent or equivalently, there is $\cQ(\cdot)=\neg\cP(\cdot)$:
		$$\cQ^\ell(\cC[0:N])=
		\begin{cases}
			False \text{, $\exists \cC^\prime \subset \cC[0:N]$  s.t. $ \cD^\ell(\cC^{\prime})=True$} \\
			True \text{, otherwise}\\
		\end{cases}$$
	
	\noindent where $\cC^\prime$ is a subset of the blockchain $\cC[0:N]$, 
	and $\cD^\ell(\cdot)$ is a computable predicate taking $\cC^\prime$ as input and is writable as:
	$$\cD^\ell(\cC^\prime)=
		\begin{cases}
			True  \text{,  $\exists{\ }\{\tx_i\}$ that \text{$|\{\tx_i\}|\leq \ell $}: }\\
			{\ }{\ }{\ }{\ }{\ }{\ }{\ }{\ }    \text{$f(\{\tx_i\})=1 ~\wedge~ \forall~ \tx_i \in \{\tx_i\},$}  \\
			{\ }{\ }{\ }{\ }{\ }{\ }{\ }{\ }{\ }{\ }{\ }{\ }{\ }{\ } \text{$ \exists{\ } \cC[t] \in \cC'$ and P.P.T. computable $\pi_i$ s.t.} \\
			{\ }{\ }{\ }{\ }{\ }{\ }{\ }{\ }{\ }{\ }{\ }{\ }{\ }{\ }{\ }{\ }{\ }{\ }{\ }{\ } \text{$\mathsf{VrfyMTP}(\cC[t].\cR,\cH(\tx_i), \pi_i)=1$}  \\
			False \text{, otherwise}\\
		\end{cases}$$
	\noindent where $f(\{\tx_i\})=1$  captures that   $\{\tx_i\}$ satisfies a certain relationship, 
	e.g., ``the hash of each $\tx_i$ equals a specified identifier $\txid_i$'',   
	or ``each $\tx_i$ can pass the membership test of a given bloom filter'', 
	or ``the overall inflow of $\{\tx_i\}$ is greater than a given value''.
	We   let $\cP^\ell_N$ and $\cQ^\ell_N$  be short for $\cP^\ell(\cC[0:N])$ and $\cQ^\ell(\cC[0:N])$, respectively.

\smallskip
\noindent{\bf Examples of chain predicate.} The seemingly complicated definition of chain predicate actually has rather straightforward intuition  to capture a wide range of  blockchain ``readings'', as for any predicate under this category, either its trueness or its falseness can be succinctly attested by up to $\ell$ transactions' inclusion in the  chain. For $\ell=1$, some concrete examples are:
\begin{itemize}  [leftmargin=0.2in]
	\item ``A certain transaction $\tx$   is included in $\cC[0:N]$'',  the trueness which can be attested by $\tx$'s inclusion in the  chain.
	\item ``A set of transactions $\{\tx_j\}$ are \emph{all} incoming transactions sent to a particular address in $\cC[0:N]$'', the falseness of which can be proven, if $\exists$ a transaction $\tx$ s.t.: (i) $\tx \notin \{\tx_j\}$, (ii) $\tx$ is sent to the certain address, and (iii) $\tx$ is included in the chain $\cC[0:N]$. 
\end{itemize}

\smallskip
\noindent  \emph{Limits}. 
A chain predicate is   a binary question, whose trueness (or falseness) is reducible to the inclusion of some transactions. 
Nevertheless, its actual meaning  depends on how to concretely  specify it.
Intuitively,  a ``meaningful'' chain predicate  might need  certain specifications
  from an external party outside the system.
%
For example, the cashier of a pizza store     can specify a transaction to evaluate its (non)existence, 
only if the customer   tells  the $\txid$.




\smallskip
\noindent{\bf ``Handicapped'' verifiability of chain predicate.}  
%
%
W.lo.g., we will focus on the chain predicate in form of $\cP^\ell_N$, namely, whose trueness is provable instead of the falseness for presentation simplicity. Such the  ``handicapped'' verifiability can be well abstracted through a tuple of two algorithms $(\evaluate, \validate)$:
\begin{itemize}
	\item $\evaluate(\cP^\ell_N) \rightarrow$  $\sigma$ or $\bot$: The algorithm  
	   takes the replica of the blockchain as auxiliary input  and   outputs $\sigma$ or $\bot$, 
	where $\sigma$ is a proof for $\cP^\ell_N=True$, and $\bot$ represents its falseness; note the proof $\sigma$ here includes: a set of transactions $\{\tx_i\}$, a set of Merkle proofs $\{\pi_i\}$, and a set of blocks $\cC'$;
	\item $\validate(\sigma, \cP^\ell_N) \rightarrow 0$ or $1$: This algorithm   
	    takes    $\cHashes$ as auxiliary input  and   outputs 1 ($\accept$) or 0 ($\reject$) depending on whether $\sigma$ is deemed to be a valid proof for  $\cP^\ell_N=True$; note the validation   parses $\sigma$ as $(\{\tx_i\}, \{\pi_i\}, \cC')$  and verifies: (i)   $\cC'$ is included by $\cHashes[t]$ where $t \le N$; (ii) each $\tx_i$ is committed by a block in $\cC'$ due to Merkle proof $\pi_i$; (iii) $f(\{\tx_i\})=1$  where $f(\cdot)$ is the specification of the chain predicate. 
\end{itemize}

\begin{center}
	
	\tikzset{every picture/.style={line width=0.75pt}} 
	
	\hspace{1cm}
	\begin{tikzpicture}[x=0.75pt,y=0.75pt,yscale=-1,xscale=1]
	
	\draw   (80,120) -- (160,120) -- (160,140) -- (80,140) -- cycle ;
	\draw   (230,120) -- (310,120) -- (310,140) -- (230,140) -- cycle ;
	\draw    (120,100) -- (120,118) ;
	\draw [shift={(120,120)}, rotate = 270] [fill={rgb, 255:red, 0; green, 0; blue, 0 }  ][line width=0.08]  [draw opacity=0] (7,-4) -- (0,0) -- (7,4) -- (6,0) -- cycle   ;
	
	\draw    (270,100) -- (270,118) ;
	\draw [shift={(270,120)}, rotate = 270] [fill={rgb, 255:red, 0; green, 0; blue, 0 }  ][line width=0.08]  [draw opacity=0] (7,-4) -- (0,0) -- (7,4) -- (6,0) -- cycle   ;
	
	\draw    (160,124) -- (228,124) ;
	\draw [shift={(230,124)}, rotate = 180] [fill={rgb, 255:red, 0; green, 0; blue, 0 }  ][line width=0.08]  [draw opacity=0] (7,-4) -- (0,0) -- (7,4) -- (6,0) -- cycle   ;
	
	\draw    (160,138) -- (198,138) ;
	\draw [shift={(200,138)}, rotate = 180] [fill={rgb, 255:red, 0; green, 0; blue, 0 }  ][line width=0.08]  [draw opacity=0] (7,-4) -- (0,0) -- (7,4) -- (6,0) -- cycle   ;
	
	\draw    (310,130) -- (328,130) ;
	\draw [shift={(330,130)}, rotate = 180] [fill={rgb, 255:red, 0; green, 0; blue, 0 }  ][line width=0.08]  [draw opacity=0] (7,-4) -- (0,0) -- (7,4) -- (6,0) -- cycle   ;
	
	\draw    (120,160) -- (120,142) ;
	\draw [shift={(120,140)}, rotate = 450] [fill={rgb, 255:red, 0; green, 0; blue, 0 }  ][line width=0.08]  [draw opacity=0] (7,-4) -- (0,0) -- (7,4) -- (6,0) -- cycle   ;
	
	\draw    (270,160) -- (270,142) ;
	\draw [shift={(270,140)}, rotate = 450] [fill={rgb, 255:red, 0; green, 0; blue, 0 }  ][line width=0.08]  [draw opacity=0] (7,-4) -- (0,0) -- (7,4) -- (6,0) -- cycle   ;
	
	\draw    (120,160) -- (270,160) ;

	\draw (120,130) node  [font=\scriptsize]  {$\evaluate$};
	\draw (270,130) node  [font=\scriptsize]  {$\validate$};
	\draw (120,96.5) node  [font=\scriptsize]  {blockchain replica};
	\draw (270.5,96.5) node  [font=\scriptsize]  {hashes of blocks};
	\draw (195,120) node  [font=\scriptsize]  {$true:\sigma $};
	\draw (182,132) node  [font=\scriptsize]  {$false:\bot $};
	\draw (342,126.5) node  [font=\scriptsize]  {$0/1$};
	\draw (197.5,153.5) node  [font=\scriptsize]  {$P^{\ell }_N$};

	\end{tikzpicture}
	
\end{center}

The above algorithms satisfy:
(i) {\em Correctness}. 
For any chain predicate $\cP^\ell_N$, there is
$\Pr[\validate(\evaluate(\cP^\ell_N),$  $\cP^\ell_N)=1 \mid \cP^\ell_N=true]=1$,
and (ii)  {\em Verifiability}. for any P.P.T.  $\hA$ and   $\cP^\ell_N$, there is
$\Pr[\validate( \sigma\leftarrow\hA(\cP^\ell), $ $\cP^\ell_N)=1  \mid \cP^\ell_N=false]\le negl(\lambda)$, 
where $\evaluate$ implicitly takes the blockchain replica as  input, and
$\validate$  implicitly inputs  $\cHashes$.
The abstraction can also be  slightly adapted for $\cQ^\ell_N$ whose falseness is the provable side, though we omit that   for presentation simplicity.
Through the remaining of the paper, $\evaluate$ can be seen as a black-box callable by any full nodes that have the complete replica of the blockchain,
and $\validate$ is a   subroutine that can be invoked by the smart contracts that can access the dictionary $\cHashes$.

%
%

\subsection{System \& adversary model}

The system explicitly consists of a light client, some relay(s) and an arbiter contract.
All of them are computationally bounded to perform only polynomial-time computations. 
The messages between them can deliver synchronously within a-priori known delay $\Delta T$, via point-to-point channels. In details,


\smallskip
{\bf The rational lightweight client} $\mathcal{LW}$ is abstracted as:
\begin{itemize} 
	\item It is {\bf rational} and selfish;
	\item It is computationally bounded, i.e.,  it can only take an     action computable in probabilistic polynomial-time; 
	\item It opts out of  consensus; to capture this, we assume: 
	\begin{itemize}
		\item The client can {\em send} messages to the   contract due to the   network diffusion functionality  \cite{GKL15,KMS16};
		\item The client cannot receive messages from the contract  except a short setup phase,
		which can be done in practice because the client user   can temporarily boost a personal full node by fast-bootstrapping protocols.
	\end{itemize}
\end{itemize}


%
{\bf The rational full node} $\mathcal{R}_i$   is modeled as:
\begin{itemize} 
	\item It is {\bf rational}. Also, the full node $\mathcal{R}_i$ might (or might not) cooperate with another full node $\mathcal{R}_j$. The {\bf (non)-cooperation} of them is specified as:
	\begin{itemize}
		\item The {\bf cooperative} full nodes form a coalition to maximize the total utility, as they can share all information, coordinate all actions  and transfer payoffs, etc. \cite{osborne1994course}; essentially, we  follow  the conventional   notion to view   the  cooperative relays as {\bf a single} party \cite{cryptoeprint:2011:396}.
		\item  {\bf Non-cooperative} full nodes  maximize  their  own utilities independently in a selfish manner due to the  standard   non-cooperative game theory,
		which can be understood as that they are not allowed to choose  some ITMs to communicate with each other \cite{lepinksi2005collusion}; 
	\end{itemize}
	\item It can only take P.P.T. computable actions at any stage of the protocol;
	\item The   full node  runs the  consensus, such that:
	\begin{itemize}
		\item It stores the complete replica of the latest blockchain;
		\item It can send/receive messages to/from the smart  contract;

	\end{itemize}
	\item It can send messages to the light client via an off-chain private channel.\footnote{Such the assumption can be granted if considering the   client and the relays can set up private communication channels  on demand. In practice, this can be done because (i) the   client can ``broadcast'' its   network address via the blockchain \cite{elastico},  or (ii) there is a trusted   name service that tracks the network addresses of the relays.}
\end{itemize}

{\bf The arbiter contract} $\mathcal{G}_{ac}$   follows the standard abstraction of smart contracts \cite{KMS16,KZZ16},
with a few slight extensions. 
First, it would not send any messages to  the light client except during a short setup phase.
Second, it can access a   dictionary $\cHashes$ \cite{blockhash,EIP210}, which contains the hashes of all blocks. 
The latter abstraction allows the contract to invoke $\validate$ to verify the proof attesting the trueness of any predicate $\cP^\ell_N$, in case the predicate is actually true.


%

%

\subsection{Economic factors of our model}

It is necessary to clarify the economic parameters of the {\bf rational} parties to complete our game-theoretic model. We present those economic factors and argue the rationale behind them  as follows:
\begin{itemize}  
	
	\item $c$: It represents how much the client spends to maintain its (personal) trusted full node.
	Note $c$ does not mean the security relies on a trusted full node and  only characterizes the cost of maintaining the trusted full node.
	For example,   $c \rightarrow \infty$  would represent that no available trusted full node. Note $c$ does not characterize the cost of relay to maintain a relaying full node. 
	
	\item $v$: The factor means the ``value'' attached to the chain predicate under query. If the client incorrectly  evaluates the    predicate, it loses $v$.
	%
	%
	For example, the  cashier Alice is evaluating the (non)existence of a certain transaction;
	if Alice believes the existence of a non-existing transaction, she loses the  amount to be transacted;
	if Alice  believes the nonexistence of an existing transaction, her business is   harmed by such the censorship.

	\item $v_i(\cP^\ell_N, \cC) \rightarrow [0,v_i]$: This function characterizes the motivation of the relay $\mathcal{R}_i$  to cheat the light client.
	Namely, it represents the extra (malicious) utility that the relay   $\mathcal{R}_i$ earns, if   fooling the   client to incorrectly evaluate the chain predicate. 
	%
	%
	We explicitly let  $v_i(\cP^\ell_N, \cC)$ have an upper-bound  $v_i$.
	
	\item $\epsilon$: When a party chooses a  strategy (i.e., a  P.P.T. ITM)  to break underlying cryptosystems,
	we let $\epsilon$   represent the expected utility  of such a strategy, where   $\epsilon$ is a negligible function in     cryptographic security parameter   \cite{DHR00}. 
	\item In addition,   all communications and P.P.T. computations can be done costlessly (unless otherwise specified).
\end{itemize}


\subsection{Security goal}

The aim of the light-client protocol in the  game-theoretic model is to allow a rational light client employ some rational relaying full nodes (e.g., two) to correctly {\em evaluate} a few chain predicates, and these recruited full nodes are correctly paid as pre-specified.
%
%
%
%
In details, we require such the  light-client protocol $\Pi_{\mathcal{LW}}$ to satisfy the following \emph{correctness} and \emph{security} properties:
\begin{itemize}  
	\item   \underline{\smash{\em Correctness}}. If all parties are honest, we require: (i) the relay nodes are correctly paid; (ii)
	the light client correctly evaluates some chain predicates under the category of $\cP^\ell(\cdot)$, regarding the chain $\cC[0:T]$ (i.e. the chain
	at the time of evaluating). Both   requirements shall hold with  probability 1.
	
	\item   \underline{\smash{\em Security}}.   We adopt a strong game-theoretic security notion of \emph{sequential equilibrium} \cite{HPS16,DWA17,halpern2019sequential}  for  incomplete-information extensive games. 
	Consider   an extensive-form game $\Gamma$ that models the light-client protocol $\Pi_{\mathcal{LW}}$,
	and let $(\bf Z_{bad},  Z_{good})$ as a partition of the terminal histories $\bf Z$ of the game $\Gamma$.
	Given a  $\epsilon$-\emph{sequential equilibrium} of $\Gamma$ denoted by $\sigma$, the probability of reaching each terminal history $z \in \bf Z$ can be induced, which can be denoted by $\rho(\sigma, z)$.
	Our security goal would require: there is a $\epsilon$-sequential equilibrium $\sigma$ of $\Gamma$ where $\epsilon$ is at most a negligible function in cryptographic security parameter $\lambda$,
	such that under the $\epsilon$-equilibrium $\sigma$, the game $\Gamma$ always terminates in $\bf Z_{good}$.
	%
	
\end{itemize}

\noindent
{\em Remark}. 
The traditional game-theory analysis captures only computationally unbounded players. 
But it becomes natural to consider computationally-bounded players in an interactive protocol using cryptography, so will we do through the paper.
%
%
%
In such the setting, 
a strategy of a party can be a P.P.T. ITM to break the underlying cryptosystems. 
However, this strategy succeeds with only negligible probability. 
Consequently, our security goal (i.e., $\epsilon$-sequential equilibrium) can be refined into a computational variant to state the rational players switch strategies, only if the gain of deviation is non-negligible.


\section{A simple light-client protocol}
\label{protocol}

We carefully design a simple  light-client protocol, in which
a light client ($\mathcal{LW}$) can leverage it to employ  two (or one) relays
to evaluate the chain predicates   $\cP^\ell_N$ (as defined in Section \ref{sec:predicate}).


%

\subsection{Arbiter contract \&   high-level of the protocol}

The simple light-client protocol is centering around an arbiter smart contract $\mathcal{G}_{ac}$ as shown in Fig \ref{fig:contract}.
It begins with letting all parties place their initial deposits in the arbiter   contract $\mathcal{G}_{ac}$.
%
%
Later, the client can ask the relays to forward some readings about the blockchain, and then feeds what it receives  back to the contract.
As such, once the contract hears the feedback from the client, it can leverage the initial deposits to facilitate some proper incentive mechanism, in order to prevent the parties from deviating by rewards and/or punishments, which becomes the crux  of our  protocol.

For security in the rational setting, the incentive mechanism must be powerful enough to   precisely punish misbehaviors (and reward honesty).
Our main principle to realize such the   powerful incentive  is letting the arbiter contract  to learn as much as possible regarding how the protocol is actually executed off-chain, so it can precisely punish and then deter any deviations.

Nevertheless, the contract has ``handicapped'' abilities.
So we have to carefully design the protocol to circumvent its limits, for the convenience of designing the powerful enough incentive mechanism   later.

%
%

\begin{figure}[!]
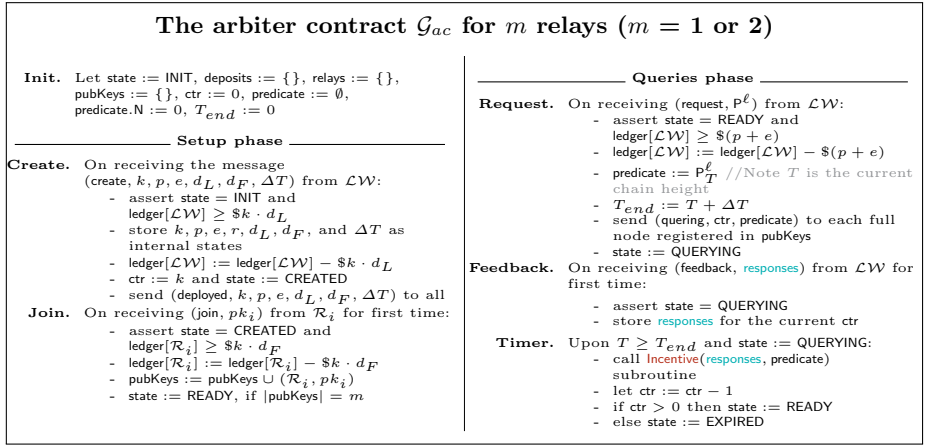

	\centering

	\fbox{%
		\parbox{0.98\linewidth}{%
			\vspace{-3mm}
			\begin{center}
				{\bf The arbiter contract $\mathcal{G}_{ac}$ for $m$ relays ($m$ = 1 or 2)} 
			\end{center}
			\vspace{-5mm}

			\begin{multicols}{2}
			{\tiny			
				\begin{flushleft}
 
					\begin{itemize}[leftmargin=0.8cm]
						\item[{\bf Init.}] Let $\mathsf{state}:=\mathsf{INIT}$, $\mathsf{deposits}:=\{\}$, $\mathsf{relays}:=\{\}$, $\mathsf{pubKeys}:=\{\}$, $\mathsf{ctr}:=0$, $\mathsf{predicate}:=\emptyset$, $\mathsf{predicate.N}:=0$, $T_{end}:=0$
					\end{itemize}
					\vspace{2mm}
										
					\hrulefill{} {\bf Setup phase} \hrulefill{}
					\vspace{-2mm}
					\begin{itemize}[leftmargin=0.92cm]
						\item[{\bf Create.}]  On receiving the message $(\mathsf{create}, k, p, e, d_{L}, d_{F}, \Delta T)$ from $\mathcal{LW}$:
						\begin{itemize}
							\item[-] assert $\mathsf{state}=\mathsf{INIT}$ and $\cLedger[\mathcal{LW}] \ge \$k\cdot d_{L}$
							\item[-] store $k, p, e, r, d_{L}, d_{F},$ and $\Delta T$ as internal states
							\item[-] $\cLedger[\mathcal{LW}]:=\cLedger[\mathcal{LW}]-\$k\cdot d_{L}$
							\item[-] $\mathsf{ctr}:=k$ and $\mathsf{state}:=\mathsf{CREATED}$ 
							\item[-] send $(\mathsf{deployed}, k, p, e, d_{L}, d_{F}, \Delta T)$ to all
						\end{itemize}
						
						\item[{\bf Join.}] On receiving $(\mathsf{join}, pk_i)$ from $\mathcal{R}_i$ for first time:
						\begin{itemize}
							\item[-] assert $\mathsf{state}=\mathsf{CREATED}$ and $\cLedger[\mathcal{R}_i] \ge \$k\cdot d_{F}$
							\item[-] $\cLedger[\mathcal{R}_i]:=\cLedger[\mathcal{R}_i]-\$k\cdot d_{F}$
							\item[-] $\mathsf{pubKeys} := \mathsf{pubKeys} \cup (\mathcal{R}_i,pk_i)$
							\item[-] $\mathsf{state}:=\mathsf{READY}$, if $|\mathsf{pubKeys}|=m$
						\end{itemize}
					\end{itemize}
					
					\columnbreak
					
					\vspace{-2mm}
					\hrulefill{} {\bf Queries phase} \hrulefill{}
					\vspace{-2mm}
					\begin{itemize}[leftmargin=1.2cm]
						
						\item[{\bf Request.}]  On receiving $(\mathsf{request}, \cP^\ell)$ from $\mathcal{LW}$:
						\begin{itemize}
							\item[-] assert $\mathsf{state}=\mathsf{READY}$ and $\cLedger[\mathcal{LW}] \ge \$(p+e)$
							\item[-] $\cLedger[\mathcal{LW}]:=\cLedger[\mathcal{LW}]-\$(p+e)$
							\item[-] $\mathsf{predicate} := \cP^\ell_T$   {\color{Gray} //Note  $T$ is the current chain height}
							\item[-] $T_{end}:=T+\Delta T$
							\item[-] send $(\mathsf{quering}, \mathsf{ctr},   \mathsf{predicate})$ to each full node registered in $\mathsf{pubKeys}$
							\item[-] $\mathsf{state}:=\mathsf{QUERYING}$
						\end{itemize}
						
						\item[{\bf Feedback.}] On receiving $(\mathsf{feedback}$, {\color{TealBlue}$\mathsf{responses}$}$)$ from $\mathcal{LW}$ for first time:
						
						\begin{itemize}
							\item[-] assert $\mathsf{state}=\mathsf{QUERYING}$
							\item[-]  store {\color{TealBlue}$\mathsf{responses}$} for the current $\mathsf{ctr}$  
							
							



						\end{itemize}
						
						\item[{\bf Timer.}] Upon $T \geq  T_{end}$ and $\mathsf{state}:=\mathsf{QUERYING}$:
						\begin{itemize}
							\item[-] call {\color{BrickRed}$\mathsf{Incentive}$}$(${\color{TealBlue}$\mathsf{responses}$}$, \mathsf{predicate})$ subroutine
							\item[-] let $\mathsf{ctr}:=\mathsf{ctr}-1$
							\item[-] if $\mathsf{ctr} > 0$ then  $\mathsf{state}:=\mathsf{READY}$
							\item[-] else $\mathsf{state}:=\mathsf{EXPIRED}$
						\end{itemize}
						
					\end{itemize}

				\end{flushleft}
			}
			\end{multicols}
			\vspace{-3mm}
			
		}%
	}
	\caption{The   contract $\mathcal{G}_{ac}$ by     pseudocode notations  in \cite{KMS16}. 
		The {\color{BrickRed}$\mathsf{Incentive}$} subroutine is decoupled from the protocol and will be presented separately in the later section.
	}\label{fig:contract}
	
\end{figure}

First,  the contract $\mathcal{G}_{ac}$ does not know what the relay nodes forward to  the light client off-chain.
So the contract $\mathcal{G}_{ac}$ has to rely on the   client to know what the relays did.
%
At the first glance, the   client might cheat the contract, by claiming that it receives nothing from the relays or even forging the relays' messages, in order to avoid paying.
To deal with the issue, 
we require that:
(i) the relays authenticate what they forward to the client by digital signatures,
so the contract later can verify whether a message was originally sent from the relays, by checking the attached signatures;
(ii) the contract requires the light client to deposit an amount of $\$e$ for each query, which is returned to the   client, only if the client reports some forwarded blockchain readings signed by the relays.

Second, the contract has a ``handicapped'' verifiability, which allows it to efficiently verify a claim of $\cP^\ell_N=True$, if being give a succinct proof
$\sigma$. To leverage the property, the protocol is designed to let the relays  attach the corresponding  proof $\sigma$ whenever claiming the provable trueness.
Again, such the design is a simple yet still useful way to allow the contract ``learn'' more about the protocol execution, which later allows us to design powerful incentive mechanisms to precisely punish deviations.



\subsection{The light-client protocol}

In the presence of the contract $\mathcal{G}_{ac}$, our light-client protocol can be formally described as Fig  \ref{fig:protocol}. 
To make an oversimplified summary, it first comes with a one-time setup phase, during which the relay(s) and client make initial deposits, which later can be leveraged by the incentive mechanism to fine tune the payoffs. 
Then, the   client can  work independently and  request the relays to  evaluate a few chain predicates up to $k$ times, repeatedly. Since the payoffs are well adjusted, ``following the protocol'' becomes the rational choice of everyone in each query.

\begin{figure}[!]
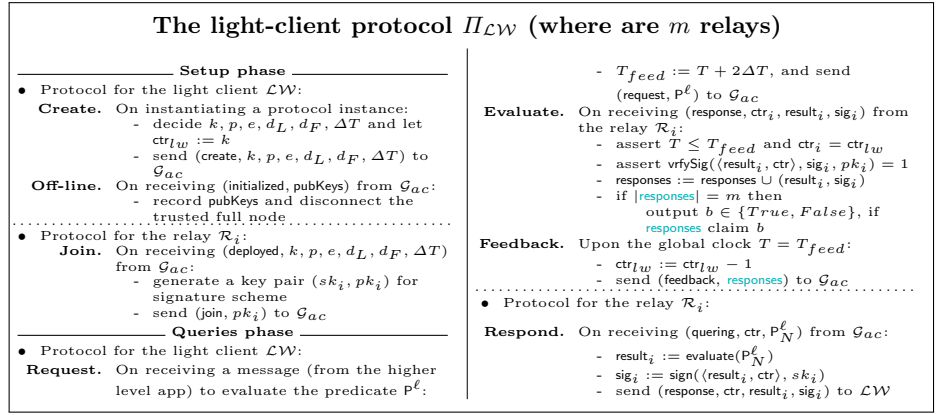

	\centering
	
	\fbox{%
		\parbox{.98\linewidth}{%
			\vspace{-3mm}
			\begin{center}
				{\bf The light-client protocol $\Pi_{\mathcal{LW}}$ (where are $m$ relays)}
			\end{center}
			\vspace{-6mm}
			
			\begin{multicols}{2}
			{\tiny						
				\begin{flushleft}
					
					\vspace{-2mm}
					\hrulefill{} {\bf Setup phase} \hrulefill{}
					\vspace{-2mm}
					\begin{itemize}[leftmargin=0.3cm]
						
						\item Protocol for the light client $\mathcal{LW}$: 
						
						\begin{itemize}[leftmargin=1.0cm]
							\item[{\bf Create.}]  On  instantiating a protocol instance:
							\begin{itemize}[leftmargin=0.5cm]
								\item[-] decide $k, p, e,  d_{L}, d_{F},  \Delta T$ and let $\mathsf{ctr}_{lw}:=k$
								\item[-] send $(\mathsf{create}, k, p, e,  d_{L}, d_{F}, \Delta T)$ to $\mathcal{G}_{ac}$
							\end{itemize}
							
							\item[{\bf Off-line.}]  On receiving   $(\mathsf{initialized}, \mathsf{pubKeys})$ from $\mathcal{G}_{ac}$:
							\begin{itemize}[leftmargin=0.5cm]
								\item[-] record $\mathsf{pubKeys}$ and disconnect   the trusted full node
							\end{itemize}
						\end{itemize}
						
						
						\vspace{-2mm}
						\hspace{-0.4cm}\dotfill{}
						
						\item Protocol for the relay $\mathcal{R}_i$:
						\begin{itemize}[leftmargin=1.0cm]
							\item[{\bf Join.}]  On receiving   $(\mathsf{deployed}, k, p, e, d_{L}, d_{F}, \Delta T)$ from $\mathcal{G}_{ac}$:
							\begin{itemize}[leftmargin=0.5cm]
								\item[-] generate a key pair $(sk_i,pk_i)$ for signature scheme
								\item[-] send $(\mathsf{join}, pk_i)$ to  $\mathcal{G}_{ac}$
							\end{itemize}
							
						\end{itemize}
						
					\end{itemize}
					
					\vspace{-2mm}
					\hrulefill{} {\bf Queries phase} \hrulefill{}
					\vspace{-2mm}
					\begin{itemize}[leftmargin=0.3cm]
						
						\item Protocol for the light client $\mathcal{LW}$: 	
						
						\begin{itemize}[leftmargin=1.0cm]
							\item[{\bf Request.}]  On receiving a message (from the higher level app) to evaluate the predicate $\cP^\ell$:
							\begin{itemize} [leftmargin=0.5cm]
								\item[-] $T_{feed}:=T+2\Delta T$, and send $(\mathsf{request}, \cP^\ell)$ to $\mathcal{G}_{ac}$
							\end{itemize}
							
							\item[{\bf Evaluate.}]  On receiving  $(\mathsf{response}, \mathsf{ctr}_i, \mathsf{result}_i, \mathsf{sig}_i)$ from the relay $\mathcal{R}_{i}$:
							\begin{itemize}[leftmargin=0.5cm]
								\item[-] assert $T \le T_{feed}$ and $\mathsf{ctr}_i=\mathsf{ctr}_{lw}$
								\item[-] assert $\vrfy(\langle \mathsf{result}_i,\mathsf{ctr} \rangle, \mathsf{sig}_i, pk_i)=1$
								\item[-] $\mathsf{responses}:=\mathsf{responses} \cup (\mathsf{result}_i, \mathsf{sig}_i)$
								\item[-] if $|${\color{TealBlue}$\mathsf{responses}$}$|=m$ then
								\begin{itemize}[leftmargin=0.4cm]
									\item[] output $b \in \{True, False\}$, if  {\color{TealBlue}$\mathsf{responses}$}   claim $b$
								\end{itemize}
							\end{itemize}
							
							\item[{\bf Feedback.}]  Upon the global clock $T = T_{feed}$:
							\begin{itemize} [leftmargin=0.5cm]
								\item[-] $\mathsf{ctr}_{lw}:=\mathsf{ctr}_{lw}-1$
								\item[-] send $(\mathsf{feedback}$, {\color{TealBlue}$\mathsf{responses}$}$)$ to $\mathcal{G}_{ac}$
								
							\end{itemize}
							
						\end{itemize}					
						
						\vspace{-2mm}
						\hspace{-0.4cm}\dotfill{}
						
						\item Protocol for the relay $\mathcal{R}_i$:
						
						\begin{itemize}[leftmargin=1.0cm]
							\item[{\bf Respond.}]  On receiving  $(\mathsf{quering}, \mathsf{ctr}, \cP^\ell_N)$ from $\mathcal{G}_{ac}$:
							\begin{itemize}[leftmargin=0.5cm]
								\item[-] $\mathsf{result}_i := \evaluate(\cP^\ell_N)$
								\item[-] $\mathsf{sig}_i := \mathsf{sign}(\langle \mathsf{result}_i,\mathsf{ctr} \rangle, sk_i)$
								\item[-] send $(\mathsf{response}, \mathsf{ctr}, \mathsf{result}_i, \mathsf{sig}_i)$ to $\mathcal{LW}$
							\end{itemize}
						\end{itemize}

					\end{itemize}
					
				\end{flushleft}	
			}
			\end{multicols}
			\vspace{-3mm}
		}
	}
	\caption{The light-client protocol $\Pi_{\mathcal{LW}}$ among  honest relay  node(s) and the light  client.}	\label{fig:protocol}
\end{figure}

	\smallskip
	\noindent{\bf Setup phase}.
	As shown in Fig \ref{fig:protocol}, the user of a lightweight client $\mathcal{LW}$ connects to a trusted full node in the setup phase,
	and announces an ``arbiter'' smart contract $\mathcal{G}_{ac}$.
	After the contract $\mathcal{G}_{ac}$ is deployed, 
	some relay full nodes (e.g. one or two) 
	are recruited to join the protocol by depositing an amount of $\$k \cdot d_F$ in the contract. 
	The public keys of the relay(s) are also recorded by contract $\mathcal{G}_{ac}$.
	
	Once the setup phase is done,  each relay full node places the initial deposits $\$k\cdot d_F$ and the light client deposit $\$k\cdot d_L$, which will be used to deter their deviations from the protocol. 
	At the same time, $\mathcal{LW}$ records the public keys of the relay(s), 
	and then disconnects the trusted full node to work independently.

	
	In practice, 
	the setup can be done by using
	many fast bootstrap methods \cite{mimblewimble,vault,KMZ17}, which allows
	the   user to efficiently launch a personal trusted full node in the PC. 
	So the light client (e.g. a smart-phone) can connect to the PC to sync.
	%
	Remark that, besides the cryptographic security parameter $\lambda$, the protocol is specified with some other    parameters:
	\begin{itemize} 
		\item $k$: 		
		The protocol is expired, after the   client   requests the relay(s) to evaluate some chain predicates for   $k$ times. 
		\item $k \cdot d_L$: This is the   deposit placed by the   client to initialize the protocol.
		\item $k \cdot d_F$: The initial deposit of a full node to join the protocol   as a relay node. 
		\item $p$: Later in each query, 
		the   client shall place this amount to cover the well-deserved payment of the relay(s). 
		\item $e$: Later in each query, the   client shall place this   deposit $e$ in addition to $p$. 

	\end{itemize}

	\noindent{\bf Repeatable query phase}.
	Once the setup is done, $\mathcal{LW}$   disconnects the trusted full node, and can ask  the relay(s) to query some chain predicates repeatedly. 
	During the queries, $\mathcal{LW}$ can   message  the arbiter contract, but cannot read the internal states of $\mathcal{G}_{ac}$.
	Informally,  each query proceeds as:
	\begin{enumerate}
		
		\item \emph{Request}. In each query, $\mathcal{LW}$ firstly sends a $\mathsf{request}$ message to the contract $\mathcal{G}_{ac}$, which encapsulates detailed specifications of a chain predicate $\cP^\ell(\cdot)$, along
		with a deposit denoted by $\$(p+e)$, where $\$p$ is the promised payment and $\$e$ is a deposit refundable only when $\mathcal{LW}$ reports what it receives from the relays. 
		%
		
		\noindent
		Once   $\mathcal{G}_{ac}$ receives the $\mathsf{request}$ message from $\mathcal{LW}$, $\mathcal{G}_{ac}$ further parameterizes the chain predicate $\cP^\ell$ as $\cP^\ell_N$, where $N \leftarrow T$ represents the current global time (i.e. the latest blockchain height). 
		%

		\item \emph{Response}. 
		The the relay full node(s)
		can learn the predicate $\cP^\ell_N$ under query (whose ground truth is fixed since $N$ is fixed and would not be flipped with the growth of the global timer $T$), and the settlement of the deposit $\$(p+e)$ by reading the arbiter contract.
		
		\noindent
		Then, the relay   node can evaluate the   predicate $\cP^\ell_N$ with  using its local blockchain replica as auxiliary input. 
		When $\cP^\ell_N=True$,  the relay   node shall send the client a $\mathsf{response}$ message including a proof $\sigma$ for trueness,
		which can be verified by the arbiter contract but not the light client;
		%
		in case $\cP^\ell_N=False$, the ``honest'' full node shall reply to the light client with a $\mathsf{response}$ message including $\bot$.
		
		\noindent
		In addition, when the relay sends a  $\mathsf{response}$ message to $\mathcal{LW}$ off-chain,\footnote{\label{note:offchain} Note that we assume the off-chain communication can be established on demand in the paper,
			which in practice can be done through a name service or ``broadcasting'' encrypted network addresses through the blockchain \cite{elastico}.}
		it also authenticates the message  by attaching its signature (which is also bounded to an increasing only counter to prevent replaying).

		\item \emph{Evaluate \& Feedback}.
		Upon receiving a $\mathsf{response}$  message    from the relay $\mathcal{R}_i$, $\mathcal{LW}$ firstly verifies that it is   authenticated by a valid signature $\mathsf{sig}_i$. If $\mathsf{sig}_i$ is valid,
		$\mathcal{LW}$ parses the $\mathsf{response}$ message to check whether $\mathcal{R}_i$ claims $\cP^\ell_N=True$ or $\cP^\ell_N=False$.
		If receiving consistent $\mathsf{response}$ message(s) from all recruited relay(s), the light client decides this consistently claimed $True$/$False$. 
		%
		%
		
		\noindent
		Then the   client     sends a $\mathsf{feedback}$ message to the    contract $\mathcal{G}_{ac}$ with containing
		these signed $\mathsf{response}$ message(s).
		
		\noindent
		Remark   we do  \emph{not} assume  the   client   follows the protocol to output and feeds back to the   contract. Instead, we  focus on proving ``following the protocol to decide an output'' is the sequential rational strategy of the   client.
		
		\item \emph{Payout}. Upon receiving the $\mathsf{feedback}$ message sent from the light client,  
		the contract $\mathcal{G}_{ac}$ shall invoke the {\color{BrickRed}$\mathsf{Incentive}$} subroutine to facilitate some payoffs.
		
		\noindent
		Functionality-wise, the payoff rules of the incentive subroutine would punish and/or reward the relay node(s) and the light client,
		such that none of them would deviate from the   protocol.
		

	\end{enumerate}

\smallskip
\noindent
\emph{Remark on correctness}. It is immediate to see the correctness: when all parties are honest, the relay(s) receive the payment pre-specified  due to incentive mechanism in the contract, and the client always outputs the ground truth of chain predicate.

\smallskip
\noindent
\emph{Remark on security}. The security   would depend on the payoffs clauses facilitated by the incentive subroutine, which will be elaborated in later subsections as we intentionally decouple the protocol and the incentive design. Intuitively, if the incentive subroutines does nothing, there is no security to any extent; since following the protocol is not any variant of equilibrium. Thus, the incentive mechanism must be carefully designed to finely tune the payoffs, in order to 
make the sequential equilibrium to be following the protocol.

\section{Adding incentives for security}\label{incentive}

Without a proper {\color{BrickRed}$\mathsf{incentive}$} subroutine,  
our simple light-client protocol   is seemingly insecure to any extent,
considering at least the relay nodes are well motivated to  cheat the client.
So this section   formally treats the light-client protocol as an extensive game, and then studies on how to squeeze most out of the ``handicapped'' abilities of the arbiter contract to
design proper incentives, such that the utility function of the game can be well adjusted to deter any party from  deviating at any stage of the protocol's extensive game.

\subsection{Challenges of designing incentives}

The main challenge  of designing proper incentives to prevent the parties from deviating is the ``handicapped'' abilities of the arbiter contract $\mathcal{G}_{ac}$: there is no proof for a claim of $\cP^\ell_N=False$, so $\mathcal{G}_{ac}$ cannot directly catch a liar who claims bogus $\cP^\ell_N=False$.
We conquer the above issue in the rational setting, 
by allowing the  contract $\mathcal{G}_{ac}$ to  believe    unverifiable claims  are correctly forwarded by rational relay(s), even if no cryptographic proofs for them.
Our solution centers around  the fact:
if a claim of $\cP^\ell_N=False$ is actually fake, 
there shall exist a succinct cryptographic proof for $\cP^\ell_N=True$, which can falsify the bogus claim of $\cP^\ell_N=False$.
As such, we derive the basic principles of designing proper incentives in different scenarios:
\begin{itemize}
	\item When there are two non-cooperative relays, we create an incentive to leverage   non-cooperative parties to audit each other, 
	so sending   fake   $\cP^\ell_N=False$ become {\em irrational} and would not happen.
	\item When there is only one relay node (which models that there are no non-cooperative relays at all), we somehow try an incentive design to let the full node ``audit'' itself,
	which means: the relay would get a higher payment, as long as it presents a verifaible claim instead of an unverifiable claim. So the relay is somehow motivated to ``audit'' itself.
\end{itemize}


\subsection{``Light-client game'' of the  protocol}

Here we present the structure of ``light-client game'' for the   simple light-client protocol presented in the earlier section. We would showcase how the extensive game does capture (i) all polynomial-time computable strategies and (ii) the incomplete information received during the course of the protocol. 

\begin{figure}
	\centering
	\includegraphics[width=8.6cm]{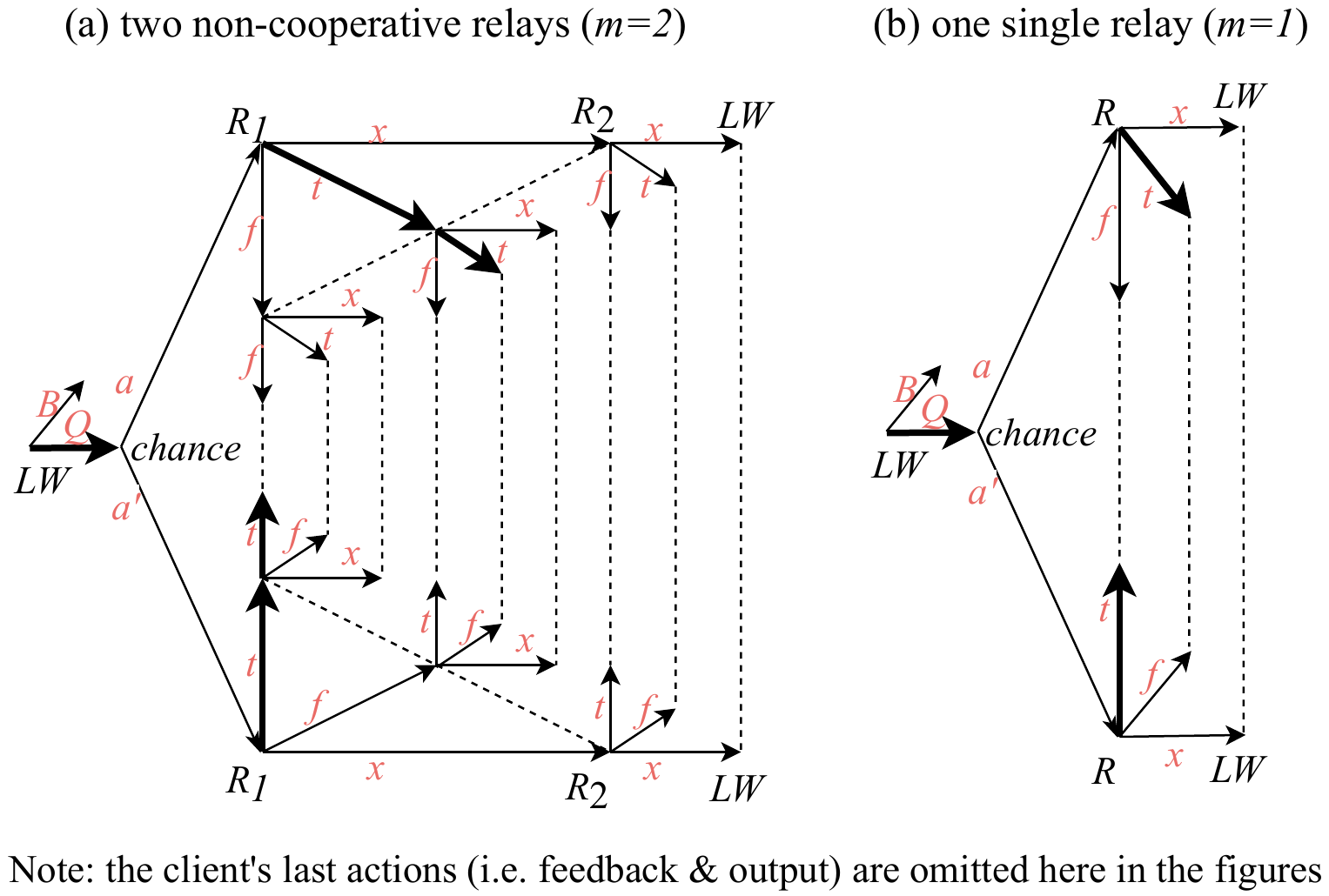}
	\caption{The repetition structure of the light-client game in one  query: (a) two non-cooperative relays   (i.e. $\Gamma_2$); (b) one single relay (i.e. $\Gamma_1$). The last actions of the client are not shown for presentation simplicity.}
	\label{fig:game}
	\vspace{-3mm}
\end{figure}

\smallskip
\noindent
{\bf Game structure for two relays}.
For the case of recruiting two (non-cooperative) relays, 
we denote the ``light-client'' game as $\Gamma^k_2$.
It has a repetition structure (i.e., a stage  game $\Gamma_2$) that can be repeated up to $k$ times as shown in Fig \ref{fig:game} (a),
since the client can raise queries for up to $k$ times in the protocol.
%
%
More precisely, for each query, the protocol proceeds as the following incomplete-information extensive stage game $\Gamma_2$:
\begin{enumerate} [leftmargin=0.55cm]
	\item \underline{\smash{\em Client makes a query}}. The client moves, with two optional actions $Q$ and $B$.
	$Q$ denotes ``sending a request message to query'', and   $B$ denotes ``others'' (including abort). 
	The game only proceeds when the light client acts $Q$.
	%

	\item \underline{\smash{\em Chance chooses the truth}}. At  the history $Q$, the special player ``$chance$'' moves, with two possible actions $a$ and $a'$.
	
	\noindent
	Let  $a$ represent $\cP^\ell_N=True$,  and   $a'$  for $\cP^\ell_N=False$. 
	The occurrence of   $a$ and $a'$ follows an arbitrary distribution $[\rho,1-\rho]$.
	Note  the action of $chance$ can be observed by the relay full nodes but not the   client.
	
	\item \underline{\smash{\em Relay responds}}. At   histories   $Q(a|a')$, \footnote{Remark that we are using standard regular expressions to denote the histories and information set. For example, $Q(a|a')$ represents $\{Qa, Qa'\}$}
	the relay node $\mathcal{R}_1$ acts, with three available actions $\{t,f,x\}$:
	\begin{itemize} 
		\item The action $t$ means   $\mathcal{R}_1$ forwards the   ground truth of $\cP^\ell_N$ to $\mathcal{LW}$ (with attaching correct ``proofs'' if there are any).\footnote{There exists another strategy to claim the truth of the predicate when the predicate is indeed true, but with invalid proof. This strategy is strictly dominated and would not be adopted at all, since it neither fools the client, nor get through the verification of contract to get any reward. We therefore omit it.}
		%
		\item The action $f$ represents that $\mathcal{R}_1$ forwards the opposite of the ground truth of chain predicate to $\mathcal{LW}$.
		%
		\item The action $x$ means as others, including abort and some attempts to break the cryptographic primitives.
	\end{itemize}
	
	\item  \underline{\smash{\em The other relay responds}}. At   histories   $Q(a|a')(t|f|x)$, it is the turn of $\mathcal{R}_2$ to move. 
	Since  $\mathcal{R}_1$ and $\mathcal{R}_2$  are non-cooperative,     histories $Qa(t|f|x)$ make of an information set of $\mathcal{R}_2$ denoted by $I_1$, and similarly, $Qa'(t|f|x)$ is another information set $I_2$.
	At either $I_1$ or $I_2$, $\mathcal{R}_2$  has three actions $\{t,f,x\}$, which can be understood as same as the actions of $\mathcal{R}_1$ at  $Q(a|a')$, since $\mathcal{R}_1$ and $\mathcal{R}_2$ are exchangeable notations.
	
	\item \underline{\smash{\em Client feeds back and outputs}}. Then the game  $\Gamma_2$ reaches the histories   $Q(a|a')(t|f|x)(t|f|x)$. As shown in Fig \ref{fig:game}, the   client $\mathcal{LW}$ is  facing nine information sets\footnote{Remark that the histories $Qatt$ and $Qa'ff$ cannot be distinguished by the light client, because for the light client, both of them correspond that two claims of $True$. All the nine information sets of the light client can be translated similarly.}:
		  $I^{\mathcal{LW}}_1=Q(att|a'ff)$, $I^{\mathcal{LW}}_2=Q(atf|a'ft)$,  $I^{\mathcal{LW}}_3=Q(at|a'f)x$,   
		  $I^{\mathcal{LW}}_4=Q(aft|a'tf)$,
		  $I^{\mathcal{LW}}_5=Q(aff|a'tt)$, $I^{\mathcal{LW}}_6=Q(af|a't)x$,
		  $I^{\mathcal{LW}}_7=Q(axt|a'xf)$, $I^{\mathcal{LW}}_8=Q(axf|a'xt)$, 
	 $I^{\mathcal{LW}}_9=Q(a|a')xx$.
		%
	\noindent
	At these information sets, the light client shall choose a probabilistic polynomial-time ITM to: (i)  send a $\mathsf{feedback}$ message back to the contract, and (ii) decide an output. 
	So the available actions of the client at each information set  can be interpreted as follows:
	\begin{itemize} 
		\item From $I^{\mathcal{LW}}_1$ to $I^{\mathcal{LW}}_5$, the client 
		receives two $\mathsf{response}$ messages from both relays in time, and it
		can take an action out of $\{T,L,R,X\}\times\{A, A', O\}$:
		$T$ means to report the arbiter contract $\mathcal{G}_{ac}$ both of the $\mathsf{response}$;
		$L$ (or $R$) represents that $\mathcal{LW}$ reports to the contract $\mathcal{G}_{ac}$ only one $\mathsf{response}$ message
		sent from $\mathcal{R}_1$ (or $\mathcal{R}_2$);
		$X$ represents others, including abort; 
		$A$ means to output $True$; $A'$ is to output $False$; 
		$O$ denotes to output nothing. 
		\item Through $I^{\mathcal{LW}}_3$ to $I^{\mathcal{LW}}_8$, the client $\mathcal{LW}$ 
		receives only one $\mathsf{response}$ message from $\mathcal{R}_1$ (or $\mathcal{R}_2$), and 
		can take an action out of $\{T,X\}\times\{A, A',O\}$:
		$T$ means to report the contract $\mathcal{G}_{ac}$ the only $\mathsf{response}$ message that it receives;
		$X$ means to do others, including abort; 
		$A$, $A'$ and $O$ have the same concrete meaning as before. 
		\item At $I^{\mathcal{LW}}_9$, $\mathcal{LW}$ receives nothing from the relays in time, it can take an action out of $\{T, X\} \times \{A, A', O\}$: $T$ can be translated as to send the contract nothing until the contract times out,
		$X$ represents others (for example, trying to crack digital signature scheme); 
		$A$, $A'$ and $O$ still have the same meaning as before. 
	\end{itemize}
	
\end{enumerate}

After all above actions are made, the protocol completes one query, and can go to the next query, as long as it is not expired or the client does not abort.
So the protocol's game $\Gamma^k_2$ (capturing all $k$ queries) can be inductively defined by repeating the above structure up to $k$ times.


\smallskip
\noindent
{\bf Game structure for one relay}.
For the case of recruiting only one relay to request up to $k$ queries, 
we denote the protocol's ``light-client'' game as $\Gamma^k_1$.
As shown in Fig \ref{fig:game} (b), it has a  repetition structure (i.e. the stage game $\Gamma_1$) similar to the game $\Gamma_2$, 
except few differences related to the information sets and available actions of the light client.
In particular,
when the client receives response from the only relay, 
it would face three information sets, namely, $Q(at|a'f)$, $Q(af|a't)$ and $Q(ax|a'x)$, instead of nine in $\Gamma_2$.
At each information set, the client always can take an action out of $\{T,X\}\times\{A, A', O\}$, 
which has the same interpretation in the game  $\Gamma_2$.
Since  $\Gamma_1$ is extremely similar to $\Gamma_2$, we omit such details here.

%

\medskip
\noindent
{\bf What if no incentive?} 
If the arbiter contract facilitates   no incentive,
the possible execution result   of the protocol can be concretely interpreted due to the economic aspects of our model as follows:
\begin{itemize}
	
	\item {\em When the client is fooled}. The client loses $\$v$, and the relay $\mathcal{R}_i$ earns  $\$v_i$. Note $\$v$ is related to the value attached to the chain predicate under query, say the transacted amount, due to our economic model; and $\$v_i$ is the malicious benefit earned by $\mathcal{R}_i$ if the client if fooled.
	\item {\em When the client outputs the ground truth}. The relay would not earn any malicious benefit, and the client would not lose any value attached to the chain predicate either.
	\item {\em When the client outputs nothing}. The relay would lose $\$c$, which means it will launch its own (personal) full node to query the chain predicate. In such case, the relay would not learn any malicious benefit. 
\end{itemize}

It is clear to see that without proper incentives to tune the above outcomes, the game cannot reach a desired equilibrium to let all parties follow the protocol, because at least the relays are well motivated to  cheat the client.
Thus we leverage the   deposits placed by the client and relay(s) to design simple yet still useful incentives in   next subsections, 
such that we can fine tune   the above outcome to realize a utility function obtaining desired equilibrium,   thus achieving security.

\subsection{Basic {\color{BrickRed}$\mathsf{incentive}$} mechanism}


%
The incentive subroutine   takes the  $\mathsf{feedback}$ message sent from the client as input, and then facilitates rewards/punishments accordingly.
After that, the       utility function of the ``light-client game'' is supposed to be well tuned  to ensure security.
Here we will present such the carefully designed incentive subroutine, 
and   analyze  the incentive makes the  ``light-client game'' secure to what extent.

\smallskip
\noindent{\bf Basic {\color{BrickRed}$\mathsf{incentive}$} for two   relays}.
If two non-cooperative relays can be recruited,
the incentive subroutine   takes the  $\mathsf{feedback}$ message   from the client as input,
and then facilitates the  incentives  following   hereunder general principles:

\begin{itemize} 
	\item It firstly  verifies whether   the $\mathsf{feedback}$ from the light client indeed encapsulates some $\mathsf{response}$s that were originally sent from     $\mathcal{R}_1$ and/or $\mathcal{R}_2$ (w.r.t. the current chain predicate under query).
	If $\mathsf{feedback}$ contains two validly signed $\mathsf{response}$s, return $\$e$ to the client;
	If $\mathsf{feedback}$ contain one validly  signed $\mathsf{response}$, return $\$e/2$ to the client.

	
	\item If a relay claims $\cP^\ell_N=True$  with attaching an invalid proof $\sigma$,
	its deposit for this query (i.e. $\$d_F$) is confiscated and would not receive any payment.

	\item When a relay sends a $\mathsf{response}$ message containing $\bot$ to claim $\cP^\ell_N=False$, there is no succinct proof attesting the claim.
	The incentive subroutine checks whether the other relay full node provides a proof attesting $\cP^\ell_N=True$. If the other relay proves $\cP^\ell_N=True$, the cheating relay loses its deposit this query (i.e. $\$d_F$) and would not receive any payment.
	%
	%
	For the other relay that falsifies the cheating claim of $\cP^\ell_N=False$, the incentive subroutine assigns it some extra bonuses (e.g. doubled payment).
	%
	

	\item After each query, if the contract does not notice a full node is misbehaving (i.e., no fake proof for truthness or fake claim of falseness), it would pay the node $\$p/2$ as the basic reward (for the honest full node). In addition, the contract returns a portion of the client's initial deposit (i.e. $\$d_L$). Moreover, the contract returns a portion of each relay's initial deposit (i.e. $\$d_F$), 
	if the incentive subroutine does not observe the relay cheats during this query. 
	
\end{itemize}

The rationale behind the   incentive design is straightforward.
First, during any   query, the rational light client will
always report  to the contract whatever the relays actually forward, since the failure of doing so always causes strictly less
utility, no matter the strategy of the relay full nodes;
Second, since the two relay full nodes are non-cooperative, 
they would be incentivized to audit each other, such that the attempt of cheating the client is deterred.
To demonstrate   above general reward/punishment  principles of the incentive mechanism are implementable, we concretely instantiate its  pseudocode  that are deferred to Appendix \ref{incentive:2}.

\smallskip
\noindent{\bf Basic {\color{BrickRed}$\mathsf{incentive}$} for one   relay}.
When any two recruited relays might collude, the situation turns to be pessimistic,
as the light client is now requesting an unknown information from only a single distrustful coalition.
To argue security in such the pessimistic case, we consider  only one relay  in the protocol. 
%
To deal with the pessimistic case, 
we tune  the incentive subroutine by incorporating the next major tuning (different from the the incentive for two relays):
\begin{itemize}
	\item If the relay claims $\cP^\ell_N=False$, its deposit is returned, but it receives a payment less than $\$p$, namely, $\$(p-r)$ where $\$r \in [0,p]$ is a parameter of the incentive.
	\item Other payoff rules are same to the basic incentive mechanism for two non-cooperative relays. 
\end{itemize}
To demonstrate the above delicately tuned  incentive is implementable,
we showcase its   pseudocode that is deferred to Appendix \ref{incentive:1-1}.



\subsection{Security analysis for basic {\color{BrickRed}$\mathsf{incentive}$}}\label{sec:basic-1}

 \smallskip
\noindent{\bf Utility function}. Putting the  financial outcome of protocol executions together with the incentive mechanism,
we can eventually derive the utility functions of   game $\Gamma_2^k$  and   game $\Gamma_1^k$, inductively.  
The formal definitions of   utilities are deferred to Appendix \ref{append:game}. 
Given such utility functions,
we can precisely analyze the light-client game $\Gamma_2^k$ (and the game $\Gamma_1^k$) to precisely understand our light-client protocol is secure to what extent.

\smallskip
\noindent{\bf Security theorems of basic {\color{BrickRed}$\mathsf{incentive}$}}.
For the case of two non-cooperative relays, the security   can be   abstracted as Theorem \ref{the:2}:

\begin{theorem}\label{the:2}
	If  the relays that join the protocol are non-cooperative,  
	there exists a $negl(\lambda)$-sequential equilibrium of $\Gamma^k_2$ that can ensure the game  $\Gamma^k_2$ terminates in a terminal history belonging $\mathbf{Z_{good}}:=(QattTA|Qa'ttTA')\{k\}$ (i.e. no deviation from the protocol), conditioned on $d_F + p/2 >   v_i $, $d_L > (p+e)$. In addition, the rational  client and the rational  relays would collectively instantiate the protocol, if $c>p$ and $p > 0$.
\end{theorem}

For the case of one single relay (which models cooperative relays), the security of the basic incentive mechanism can be   abstracted as:

\begin{theorem}\label{the:1}
	In the pessimistic case where is only one single relay (which models a coalition of relays),  there exists a  $negl(\lambda)$-sequential equilibrium that can ensure the game $\Gamma^k_1$ terminates in   $\mathbf{Z_{good}}:=(QatTA|Qa'tTA')\{k\}$ (i.e. no deviation from the protocol), conditioned on $d_F + p-r >   v_i $, $r>v_i $, and $d_L > (p+e)$. Moreover, the rational  client and the rational relay would collectively set up the protocol, if $c>p$, $p - r > 0$ and $p > 0$.
\end{theorem}

\noindent{\bf Interpretations of    Theorem \ref{the:2}}. 
The  theorem reveals that: conditioned on there are two non-cooperative relays, the sufficient conditions of security are: 
(i) the initial deposit $d_F$ of   relay node is greater than its malicious benefit $v_i$ that can be obtained by fooling the client;
(ii) the initial deposit $d_L$ of the   client is greater than the payment $p$  plus another small parameter $e$;
(iii) for the light client, it.
The above conclusion essentially hints us how to safely set up the light-client protocol to instantiate a cryptocurrency wallet in practice, 
that is: let the light client and the relays finely tune and specify their initial deposits, 
such that the client can query the (non)existence of any transaction, as long as the transacted amount of the transaction is not greater than the initial deposit placed by the relay nodes.

\smallskip
\noindent{\bf Interpretations of    Theorem  \ref{the:1}}. 
The  theorem states that:
even in an extremely hostile scenario where only one single relay exists, 
deviations are still prevented when fooling the light client to believe the non-existence of an existing transaction does not yield better payoff than honestly proving the existence.
The statement presents a feasibility region of our protocol that at least captures many important DApps (e.g. decentralized messaging apps) in practice, namely:
 fooling the client   is not very financially beneficial for the relay, and only brings a payoff $v_i$ to the relay;
 so as long the client prefers to pay a little bit more than $v_i$ to read a record in the blockchain, no one would deviate from the protocol.

\subsection{Augmented {\color{BrickRed}$\mathsf{incentive}$}}

This subsection further discusses the pessimistic scenario that no non-cooperative relays can be identified,
by introducing an extra   assumption that:  at least one public full node (denoted by $\mathcal{PFN}$) can monitor the internal states of the arbiter contract at  cost $\epsilon$, and does not cooperate with the only recruited relay. 
This extra rationality assumption can boost an incentive mechanism to   deter the relay and client from deviating from the light-client protocol. 
Here we present this augmented incentive design, and analyze its security guarantees.

\smallskip
\noindent{\bf Augmented {\color{BrickRed}$\mathsf{incentive}$} for one relay}.
The tuning of the incentive mechanism stems from the observation that:
if there is {\em any} public full node that does not cooperate with the recruited relay (and monitor the internal states of the arbiter contract),
it can stand out to audit a fake claim about $\cP^\ell_N=False$ by producing a proof attesting $\cP^\ell_N=True$.
Thus, we slightly tune the incentive subroutine  (by adding few lines of pseudocode), which can be summarized as:
\begin{itemize}
	\item When a relay forwards a $\mathsf{response}$ message containing $\bot$ to claim $\cP^\ell_N=False$, 
	the incentive subroutine shall wait few clock periods (e.g., one).
	During the waiting time, the public full node is allowed to send a proof attesting $\cP^\ell_N=True$ in order to falsify a fake claim of $\cP^\ell_N=False$;
	in this case, the initial deposit $d_F$ of the cheating relay is confiscated and sent to the public full node who stands out.
	\item Other payoff rules are same to the basic incentive mechanism, so do not  involve the public full node. 
\end{itemize}

We defer the formal   instantiation of the above augmented incentive mechanism to Appendix \ref{incentive:1-2}.

\subsection{Security analysis for augmented {\color{BrickRed}$\mathsf{incentive}$}}

\smallskip
\noindent{\bf Augmented ``light-client game''}. By the introduction of the extra incentive clause, the ``light-client game'' $\Gamma_1$ is extended to the augmented light-client game $G_1$. 
As shown in Fig \ref{fig:game-aug}, the major differences from the original light-client game $\Gamma_1$ are two aspects:
(i)  the public full node ($\mathcal{PFN}$) can choose to monitor the arbiter contract (denoted by $m$) or otherwise ($x$) in each query, which cannot be told by the relay node due to the non-cooperation, and (ii) when the ground truth of predicate is true, if the relay  cheats, $\mathcal{PFN}$ has an action  
``debate'' by showing the   incentive mechanism a proof attesting the predicate is   true, conditioned on having taken   action $m$.

The security intuition thus becomes clear:   if the recruited relay chooses a strategy to cheat with non-negligible probability, the best strategy of the public full node is to act $m$, which on the contrary deters the relay from cheating. In the other word, the relay at most deviate with negligible probability.

\begin{figure}
		\vspace{-3mm}
	\centering
	\includegraphics[width=8.2cm]{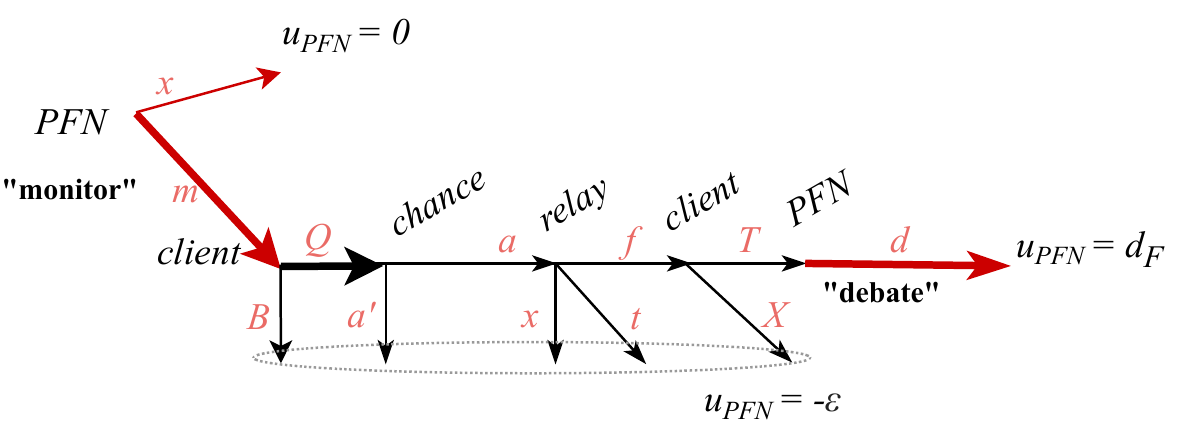}
	\caption{The induced game $G_1$, if having a non-cooperative public full node.}
	\label{fig:game-aug}
	\vspace{-3mm}
\end{figure}

\smallskip
\noindent{\bf Security theorem of augmented {\color{BrickRed}$\mathsf{incentive}$}}.
Now, in the augmented game $G_1$, if the recruited relay deviates when the predicate is true with non-negligible probability, 
the rational $\mathcal{PFN}$ would act $m$ and then $d$, which will confiscate the initial deposit of the relay and deters it from cheating.
More precisely, the security due to the augmented incentive mechanism can be summarized as:

\smallskip
\begin{theorem}\label{the:1-1}
	Given the augmented incentive mechanism, there exists a $\epsilon$-sequential equilibrium of the augmented light-client game $G_1$ such that it can ensure the client and the relay would not deviate from   the protocol except with negligible probability, conditioned on $d_F >   v_i $, $d_L > (p+e)$ and a non-cooperative public full node that can ``monitor'' the arbiter contract costlessly. Also, the   client and the relay would   set up the protocol, if $c>p$ and $p > 0$.
\end{theorem}

\smallskip
\noindent{\bf  Interpretations of    Theorem \ref{the:1-1}}. The economics behind the theorem can be   translated similarly to Theorem \ref{the:2}.

\section{Discussions on concrete instantiation}\label{instantiation}
Here  we shed   light on the concrete instantiation of the protocol in practice and emphasize some tips towards feasibility.
%
%
%
Though the current  permissionless blockchains (e.g., Ethereum) are suffering from many baby-age limitations (e.g., high cost of on-chain resources, low throughput, and large latency), a straight instantiation of our light-client protocol has  been arguably practical.

\smallskip
\noindent{\bf On-  and off-chain feasibility}. As shown in Table \ref{tab:evaluation}, we instantiate the   protocol  atop Ethereum (with recruiting one relay and using the basic incentive mechanism, c.f., Section \ref{sec:basic-1}), and measure the costs of repeatedly evaluating five chain predicates about the (non)existence of   different Ethereum transactions. 

Due to the simple nature of our protocol, the off-chain cost of the light client is constant and essentially tiny, as it only needs:
(i) to store two public keys,  (ii) to instantiate two secure
channels to connect the relay nodes (e.g. the off-chain $\mathsf{response}$  message is $<$ 1KB), (iii) to verify
two signatures and to compute a few hashes to few verify Merkle
tree proof(s) in the worst case.
%

\renewcommand{\arraystretch}{1.25}
\begin{table}[h]
	\caption{An instance of the light-client protocol (basic incentive w/ one relay)  that allows the superlight client to predicate the (non)existence up to five Ether  transactions. \label{tab:evaluation}}
	\scriptsize
	\centering
	\vspace{-5mm}
	\begin{tabular}{lllll}
		\hline
		      \begin{tabular}[c]{@{}l@{}}$\mathsf{txid}$ queried for\\evaluating (non)existence\end{tabular}                       & \begin{tabular}[c]{@{}l@{}}Gas of $\mathsf{request}$ \\ ($\mathcal{LW} \rightarrow \mathcal{G}_{ac}$) \end{tabular} & \begin{tabular}[c]{@{}l@{}}Size of  $\mathsf{response}$ \\ ($\mathcal{R}_i \rightarrow \mathcal{LW}$)\end{tabular} & \begin{tabular}[c]{@{}l@{}}Gas of $\mathsf{feedback}$ \\ ($\mathcal{LW} \rightarrow \mathcal{G}_{ac}$ )\end{tabular} \\ \hline\hline
		\href{https://etherscan.io/tx/0x14198912703598b497ffbe17e6a90ffa85de276418ae41d217c0e3d2c76290d7}{$\mathtt{0x141989127035...}$}  & 71,120 gas                                                                            & 947 Byte                                                                              & 199,691 gas                                                                            \\
		\href{https://etherscan.io/tx/0x0661d6e95ab1320e93165c37b0e666c96c97b49a8819c02a8cda05c38ff5a97d}{$\mathtt{0x0661d6e95ab1...}$} & 41,120 gas                                                                           & 951 Byte                                                                              & 251,480 gas                                                                           \\
		\href{https://etherscan.io/tx/0x949ae094deb031cbcfb1aaf36ca62b49d5e1d34affbbda16f7323568d8ac2689}{$\mathtt{0x949ae094deb0...}$} & 41,120  gas                                                                          & 949 Byte                                                                              & 257,473 gas                                                                            \\
		\href{https://etherscan.io/tx/0x1e39d5b4b46d420e960e12ba2544988bc11c4d4e8c12ceca4871d306afd73d44}{$\mathtt{0x1e39d5b4b46d...}$} & 41,120   gas                                                                         & 985 Byte                                                                              & 339,237  gas                                                                          \\
		\href{https://etherscan.io/tx/0xfe28a4dffb8ece2337863fba8dfd2686e9a9c995b50fe7911acbb2589f80b315}{$\mathtt{0xfe28a4dffb8e...}$} & 41,120    gas                                                                        & 951 Byte                                                                              & 248,119   gas                                                                         \\ \hline
	\end{tabular}
	\\	The code of testing those cases is available at \href{https://github.com/yylluu/rational-light-client}{https://github.com/yylluu/rational-light-client}. These five transactions under query are included by the Ethereum blockchain, and we choose them from the blocks having various sizes. For example, the transaction \href{https://etherscan.io/tx/0x1e39d5b4b46d420e960e12ba2544988bc11c4d4e8c12ceca4871d306afd73d44}{$\mathtt{0x1e39...}$} is in a block having 263 transactions, which indicates the evaluation has captured some worst cases of reading from large blocks. In lieu of EIP-210  \cite{EIP210} which  currently is not  available in EVM,  we hardcore  the needed blockhashes (in the contract) to measure the actual on-chain overhead (as if EIP-210 is available).
		\vspace{-2mm}
\end{table}

Besides the straightforward off-chain efficiency, the on-chain cost is also low. 
Particularly, the client only sends two messages (i.e. $\mathsf{request}$ and $\mathsf{feedback}$) to the  contract, which typically   costs mere 300k gases in the worst case as shown in Table \ref{tab:evaluation}. At the time of writing (Jan/13/2020),   ether is \$143   each \cite{ethprice}, and the average gas price is 10 Gwei \cite{ethgasstation}, which corresponds to a cost  of only \$0.43.

\medskip
\noindent{\bf Latency}. 
%
If   the network diffuse functionality \cite{GKL15}   can   approximate the latency of global Internet   \cite{BDO12,gencer2018decentralization}, 
the delay of our light-client protocol will be dominated by the limitations of underlying blockchain. 
The reasons are: (i) many existing blockchains have limited on-chain resources, and the miners are more willing to pack the transactions having higher transaction fees \cite{Nak08,Woo14,BDO12}, and (ii)  messaging the contract suffers from the intrinsic delay caused by   underlying consensus. 
For example, in Ethereum network at the time of writing,
if the light client sets its transactions at the average gas  price (i.e., 9 Gwei), the latency of messaging the contract on average will include: (i) 10 blocks (about two minutes) \cite{ethgasstation} for being mined, plus (ii)  a few more blocks for confirmations (another a few minutes) \cite{gervais2016security}.
If the   client expects the protocol to proceed faster, it can set higher gas price (e.g., 22 Gwei per gas), 
which causes its messages to be included after 2-3 blocks on average (i.e. about 30-45 seconds) \cite{ethgasstation}, though the on-chain cost    increases by 144\%.
After all, once the underlying   blockchain goes through the baby-age limitations, the   protocol's latency can be further reduced to approximate the actual Internet delay.



\medskip
\noindent{\bf Who are the relays?}
The light-client protocol can be deployed in any   blockchain supporting smart contracts. The relays in the protocol can be the full nodes of the chain (e.g., the full nodes of two competing mining pools) that are seeking   the economic  
rewards by relaying blockchain readings to the light clients, so it is reasonable to assume that they can maintain the full nodes to evaluate chain predicates nearly costlessly. 
Even in the extremely adversarial environment where the light client has no confidence in the non-collusion of any two full nodes, the protocol can still be finely tuned (e.g. increase the rewards) to support  at least  a wide range of useful low-value chain predicates.

\medskip
\noindent{\bf The initial setup}.
We explicitly decouple the presentations  of   ``protocol'' and ``incentive mechanism'' to provide the next insight:
the initial deposit is not necessary to be cryptocurrency as our design, 
and it can be any form of ``collateral'', such as business reputations and subscriptions;
especially, if the ``deposit'' is publicly known off the chain, the setup phase also becomes arguably removable, 
as the light client has no need to rely on a personal full node   to verify the correct on-chain setup of the initial deposit anymore.

\medskip
\noindent{\bf The amount of initial deposits}. One might worry that the amount of initial deposit, especially when considering that the needed initial deposit is linear to the number of queries to be asked. In practice, a few instantiations can avoid the deposit from being too large to be feasible. One of those is to let the light client and the relay node(s) to negotiate before (or during) the setup phase to choose a moderate number of queries to support, and then they can periodically reset the protocol, which is feasible as the light client user can afford to periodically reset her personal full node for a short term to handle the setups. Another possibility, as already mentioned, is relying on some external  ``collateral'' (e.g., reputations and subscriptions) to replace the deposit of cryptocurrency in the protocol.


\section{Conclusion \& Future Outlook}
We present a generic superlight client for permissionless blockchains in the game-theoretic model. 
Our protocol   corresponds a special variant of multi-party computation in the rational setting \cite{cryptoeprint:2011:396,groce2012fair,fuchsbauer2010efficient,IML05,HT04,ADGH06}, in particular, 
by adopting the smart contract to implement a concrete utility function for desired equilibrium.

This is the first work that formally discusses the light  clients of permissionless blockchains in game-theoretic settings, and the area remains largely unexplored. 
Particularly, a few potential studies can be conducted to explore more realistic instantiations. For example, many crypto-economic protocols (e.g. PoS blockchains \cite{ouroboros,snowwhite,algorand}, payment channels \cite{sprites,perun}, blockchain-backed assets \cite{ZHL19})   already introduce  many locked deposits, and it becomes enticing to explore the composability of using the same collateral in our lightweight protocol without scarifying the securities of both protocols. 



	\bibliographystyle{splncs04}
	\bibliography{reference}


    \appendix

	\newpage
	\section{Merkle Tree Algorithms}\label{append:merkle}		
	
	Fig \ref{fig:buildmt}, \ref{fig:genmerkle} and \ref{fig:vrfymerkle} are the deferred algorithms related to Merkle tree.
	A Merkle tree (denoted by $\MT$) is  a binary tree: (i) the $i$-th leaf node is labeled by $\cH(\tx_i)$, where $\cH$ is a cryptographic hash function; (ii) any non-leaf node in the Merkle tree is labeled by the hash of the labels of its siblings.
	$\mathsf{BuildMT}(\cdot)$ takes a sequence of transactions as input, and outputs a Merkle tree whose root commits these transactions. $\mathsf{GenMTP}$ takes a Merkle tree $\MT$ and a transaction $tx$ as input, and can generate a   Merkle tree proof $\pi$ for the inclusion of $tx$ in the   tree. Finally, $\mathsf{VrfyMTP}$   takes $\cH(tx)$,  the tree $\cR$, and the Merkle   proof $\pi$, and can output whether $\cH(tx)$ labels a leaf of the Merkle tree $\MT$. The security notions of Merkle tree can be found in Section \ref{preliminary}.

	\begin{figure}
		\fbox{%
			\parbox{.9\linewidth}{%
				\vspace{-3mm}
				\begin{center}
					{\bf $\mathsf{BuildMT}$ algorithm}
				\end{center}
				\vspace{-3mm}
			{\scriptsize			
				$\mathsf{BuildMT}(\TX=(tx_1,\cdots,tx_n))$:
				\vspace{-3mm}
				\begin{itemize}
					\item if $|\TX|=1$:
					\begin{itemize}
						\item $label(\cR) = \cH(tx_1)$
					\end{itemize}
					\item else:
					\begin{itemize}
						\item $lchild = \mathsf{BuildMT}(tx_1,\dots,tx_{\lceil n/2 \rceil})$
						\item $rchild = \mathsf{BuildMT}(tx_{\lceil n/2 \rceil + 1},\dots,tx_n)$
						\item $label(\cR) = \cH(label(lchild) || label(rchild))$
					\end{itemize}					
					\item return Merkle tree $MT$ with $\cR$
				\end{itemize}
		}
		\vspace{-3mm}
			}
		}
		\caption{$\mathsf{BuildMT}$ that generates a Merkle tree with $\cR$ for $\TX=(tx_1,\cdots,tx_n)$.}		\label{fig:buildmt}
	\end{figure}	
				\vspace{-3mm}

	\begin{figure}
		\fbox{%
			\parbox{.9\linewidth}{%
				\vspace{-3mm}
				\begin{center}
					{\bf $\mathsf{GenMTP}$ algorithm}
				\end{center}
			\vspace{-3mm}
				{\scriptsize			
				$\mathsf{GenMTP}(\MT, tx)$:
				\vspace{-3mm}
				\begin{itemize}
					\item $x \leftarrow$ the leaf node labeled by $\cH(tx)$
					\item while $x \neq label(\MT.\cR)$:
					\begin{itemize}
						\item $lchild \leftarrow x.parent.lchild$
						\item $rchild \leftarrow x.parent.rchild$
						\item if $x = lchild$, $b_i \leftarrow 0$, $l_i = label(rchild)$
						\item otherwise, $b_i \leftarrow 1$, $l_i = label(lchild)$
						\item $x \leftarrow x.parent$
					\end{itemize}
					\item return $\pi = ((l_i,b_i))_{i \in [1,n]}$
				\end{itemize}
				}
			\vspace{-3mm}
			}
		}
		\caption{$\mathsf{GenMTP}$ that generates a Merkle tree proof.}		\label{fig:genmerkle}
	\end{figure}
				\vspace{-3mm}

	\begin{figure}
		\fbox{%
			\parbox{.9\linewidth}{%
				\vspace{-3mm}
				\begin{center}
					{\bf $\mathsf{VrfyMTP}$ algorithm}
				\end{center}
			\vspace{-3mm}
				{\scriptsize			
				$\mathsf{VrfyMTP}(lable(\cR),\pi,\cH(tx))$:
				\vspace{-3mm}
				\begin{itemize}
					\item parse $\pi$ as a list $((l_i,b_i))_{i \in [1,n]}$  
					\item $x=\cH(tx)$
					\item for $i$ in $[1,n]$:
					\begin{itemize}
						\item if $b_i = 0$, $x \leftarrow \cH(x||l_{i})$, else $x \leftarrow \cH(l_{i}||x)$
					\end{itemize}
					\item if $x \neq lable(\cR)$, return $False$, otherwise return $True$
				\end{itemize}
				\vspace{-3mm}
				}
			}
		}
		\caption{$\mathsf{VrfyMTP}$ that verifies a Merkle tree proof.}		\label{fig:vrfymerkle}
	\end{figure}
			\vspace{-3mm}

	\section{Deferred game-theory definitions}\label{append:seq}
	
	
	Here down below are the deferred formal definitions of the finite incomplete-information extensive-form game  and the  sequential equilibrium.

	\begin{definition}
		\noindent {\bf Finite incomplete-information extensive-form game} 
		$\Gamma$ is a tuple of $\langle {\bf N,H},P, f_c, (u_i)_{i\in \bf N}, ({\bf I}_i)_{i\in \bf N} \rangle$ \cite{osborne1994course}:
		\begin{itemize}[leftmargin=0.2in] 
			\item ${\bf N}$ is a finite set representing the players.
			\item $\bf H$ is a set of sequences that satisfies: (i) $\emptyset \in \bf H$;
			(ii) if $h=\langle a_1,\dots,a_K \rangle \in \bf H$, then any prefix of $h$ belongs to $\bf H$. 
			Each member of $\bf H$ is a \emph{history} sequence. The elements of a history are called \emph{actions}. A history sequence $h=\langle a_1,\dots,a_K \rangle \in \bf H$ is \emph{terminal}, iff  $h$ is not a prefix of any other histories in $\bf H$. Let ${\bf Z}$ denote the set of terminal histories. For any \emph{non-terminal} history $h=\langle a_1,\dots,a_K \rangle \in \bf H \setminus Z$, the set of actions available after $h$ can be defined as $A(h) = \{a| \langle a_1,\dots,a_K, a \rangle \in \bf H\}$.
			\item $P:{\bf H \setminus Z} \rightarrow {\bf N} \cup \{chance\}$ is the player function to assign a player (or \emph{chance}) to move at a non-terminal history $h$. Particularly when $P(h)=chance$, a special ``player'' called \emph{chance} 
			acts at the history $h$.
			\item $(u_i)_{i\in \bf N}:{\bf Z} \rightarrow \R^{|\bf N|}$ is the utility function that defines the utility of the players at each terminal history (e.g., $u_i(h)$ specifies the utility of player $i$ at the terminal history $h$).
			\item $f_c$ is a function associating each history $h \in \{ h|P(h)=chance\}$  with a probability measure $f_c(;h)$ on $A(h)$, i.e., $f_c(a;h)$ determines the probability of the occurrence of $a \in A(h)$ after the history $h$ of the  player ``\emph{chance}''. 
			\item $({\bf I}_i)_{i\in \bf N}$ is a set of partitions. Each ${\bf I}_i$ is a partition for the set $\{h|P(h)=i\}$, and called as the \emph{information partition} of player $i$; a member $I_{i,j}$ of the partition ${\bf I}_i$ is a set of histories, and is said to be an \emph{information set} of player $i$. We require $A(h)=A(h')$ if $h$ and $h'$ are in the same information set, and then denote the available actions of player $i$ at an information set $I_{i,j} \in {\bf I}_i$ as $A(I_{i,j})$.
		\end{itemize}
	\end{definition}
	
	Note that in our context, the strategy of a player is actually a probabilistic polynomial-time ITM, whose action is to produce a string feed to the protocol. 
	
	\noindent
	\begin{definition}
		A {\em behavioral strategy} (or {\em  strategy} for short) of player $i$ (denoted by $s_i$) is a collection of independent probability measures denoted by $\{\beta_i(I_{i,j})\}_{I_{i,j}\in{\bf I}_i}$, where $\beta_i(I_{i,j})$ is a probability measure over $A(I_{i,j})$ (i.e., the available actions of player $i$ at his information set $I_{i,j}$). 
		We say $\vec{s}=(s)_{i\in{\bf {\bf N}}}$ is a {\em behavior strategy profile} (or strategy profile for short), if every $s_i \in \vec{s}$ is a behavior strategy of player $i$. When $\vec{s}=(\{\beta_i(I_{i,j})\}_{I_{i,j}\in{\bf I}_i})_{i\in{\bf {\bf N}}}$ assigns positive probability to every
		action, it is called completely mixed.
	\end{definition}

	\noindent
	\begin{definition} 
		An {\em assessment} $\sigma$ in an extensive game is a pair $(\vec{s},\mu)$, where $\vec{s}$ is a behavioral strategy profile and $\mu$ is a function that assigns to every information set a probability measure on the set of histories in the information set (i.e., every). We say the function $\mu$ is a belief system. 
	\end{definition}

	
	\noindent
	\begin{definition} 
		The {\em expected utility} of a player $i$ determined by the assessment $\sigma=(\vec{s},\mu)$ conditioned on $I_{i,j}$ is defined as:

		$$\bar{u}_i(\vec{s},\mu|I_{i,j})=\sum_{h \in I_{i,j}} \mu(I_{i,j})(h) \sum_{z \in {\bf Z}} \rho(\vec{s}|h)(z)u_i(z)$$

		\noindent where $h=\langle a_1,\dots,a_L \rangle \in {\bf H \setminus Z}$, $z=\langle a_1,\dots,a_K \rangle \in {\bf Z}$, and $\rho(\vec{\sigma}|h)(z)$ denotes the distribution over terminal histories induced by the strategy profile $\vec{s}$ conditioned on the history $h$ being reached (for player $P(h)$ to take an action), i.e.,

		$$\rho(\vec{s}|h)(z) = 
		\begin{cases}
		0 \text{, if $h$ is not a prefix of $z$}\\
		\prod_{k=L}^{K-1} \beta_{P(a_1,\dots,a_k)}(a_1,\dots,a_k)(a_{k+1})  \text{, otherwise}\\
		\end{cases}
		$$
	\end{definition}
	


	
	\noindent
	\begin{definition} 
		We say an assessment $\sigma=(\vec{\sigma},\mu)$ is a $\epsilon$-{\em  sequential equilibrium}, if it is $\epsilon$-sequentially rational and consistent:
	\end{definition}
	\begin{itemize}
		\item $(\vec{s},\mu)$ is $\epsilon$-{\em sequentially rational} if for every play $i \in {\bf N}$ and his every information set $I_{i,j} \in {\bf I}_i$, the strategy $s_i$ of player $i$ is a best response to the others' strategies $\vec{s}_{-i}$ given his belief at $I_{i,j}$, i.e., $\bar{u}_i(\vec{s},\mu|I_{i,j})+\epsilon \ge \bar{u}_i((s_{i}^\ast,\vec{s}_{-i}),\mu|I_{i,j})$ for every strategy $s_{i}^\ast$ of every player $i$ at every information set $I_{i,j} \in {\bf I}_i$. Note that   $\vec{s}_{-i}$   denotes the strategy profile $\vec{s}$ with its $i$-th element removed, and $(s^\ast,\vec{s}_{-i})$ denotes $\vec{s}$ with its $i$-th element replaced by $s^\ast$. 
		\item $(\vec{s},\mu)$ is {\em consistent}, if $\exists$ a sequence of assessments $((\vec{s}^k,\mu^k))_{k=1}^\infty$ converges to $(\vec{s},\mu)$, where  $\vec{s}^k$ is completely mixed and $\mu^k$ is derived from $\vec{s}^k$ by Bayes' rules.
	\end{itemize}

\section{Deferred formal description of incentive subroutines}\label{append:incentive}

\subsection{Basic incentive   for the protocol with two  relays}\label{incentive:2}

Fig \ref{fig:incentive2}, \ref{fig:payout} and \ref{fig:payout1} showcase the basic incentive mechanism for two relays is implementable 
with being given the $\validate$ algorithm.
Fig \ref{fig:payout} presents the detailed incentive clauses when the client feeds back two validly signed response messages from both recruited relays (clause 1-6).
Fig \ref{fig:payout1} presents how to deal with the scenario where the client only  feeds back one validly signed response message from only relay  (clause 7-9).
In case the client feeds back nothing, every party simply gets their initial deposit back (clause 10), while the locked deposit $\$(p+e)$ of the client (for this query) is ``burnt''.

\begin{figure}[H]
	\fbox{%
		\parbox{\linewidth}{%
			\vspace{-3mm}
			\begin{center}
				{\bf {\color{BrickRed}$\mathsf{Incentive}$}  subroutine} for two non-cooperative relays
			\end{center}
			\vspace{-3mm}
			{\scriptsize
			{\color{BrickRed}$\mathsf{Incentive}$}$(${\color{TealBlue}$\mathsf{responses}$}$, \cP^\ell_N)$:
			\vspace{-3mm}
			
			\begin{itemize}[leftmargin=0.4cm]

				\item[] if $|\mathsf{responses}|=2$ then 
				\begin{itemize}[leftmargin=0.4cm]
					\item[] parse $\{(\mathsf{result}_i, \mathsf{sig}_i)\}_{i \in \{1,2\}}:=\mathsf{responses}$
					\item[] if  ${\vrfy}(\langle\mathsf{result}_i, \mathsf{ctr}\rangle, \mathsf{sig}_i, pk_1)=1$  for each $i\in[1,2]$
					\begin{itemize}[leftmargin=0.4cm]
						\item[] call {\color{CadetBlue}$\mathsf{Payout}$}$(\mathsf{result}_1, \mathsf{result}_2, \cP^\ell_N)$ subroutine in Fig  \ref{fig:payout}
						\item[] $\cLedger[\mathcal{LW}]:=\cLedger[\mathcal{LW}]+\$d_L$
						\item[] return
					\end{itemize}
				\end{itemize}

				\item[] if $|\mathsf{responses}|=1$ then 
				\begin{itemize}[leftmargin=0.4cm]
					\item[] parse $\{(\mathsf{result}_i, \mathsf{sig}_i)\}:=\mathsf{responses}$
					\item[] if $\mathsf{vrfy}(\mathsf{result}_i||\mathsf{ctr}, \mathsf{sig}_i)=1$
					\begin{itemize}[leftmargin=0.4cm]
						\item[] call {\color{Lavender}$\mathsf{Payout}^\prime$}$(\mathsf{result}_i, \cP^\ell_N)$ subroutine in Fig \ref{fig:payout1} 
						\item[] $\cLedger[\mathcal{LW}]:=\cLedger[\mathcal{LW}]+\$d_L$
						\item[] return
					\end{itemize}
				\end{itemize}
			
				\item[]  {\color{Dandelion} // Clause 10}

						\item[] $\cLedger[\mathcal{R}_i]:=\cLedger[\mathcal{R}_i]+\${d_F}$
						\item[] $\cLedger[\mathcal{R}_{|1-i|}]:=\cLedger[\mathcal{R}_{|1-i|}]+\${d_F}$
						\item[] $\cLedger[\mathcal{LW}]:=\cLedger[\mathcal{LW}]+\$d_L$

			\end{itemize}
			\vspace{-3mm}
			}
		}%
	}
	\caption{{\color{BrickRed}$\mathsf{Incentive}$} subroutine (two non-cooperative relays).}\label{fig:incentive2}
 			\vspace{-3mm}
\end{figure}

\begin{figure}[H]
	\fbox{%
		\parbox{1.01\linewidth}{%
			\vspace{-3mm}	
			\begin{center}
				{\bf {\color{CadetBlue}$\mathsf{Payout}$} subroutine} 
			\end{center}
			\vspace{-3mm}
			{\scriptsize
			{\color{CadetBlue}$\mathsf{Payout}$}$(\mathsf{result}_1, \mathsf{result}_2, \cP^\ell_N)$:
			\vspace{-3mm}
			
			\begin{itemize}[leftmargin=0.4cm]
				\item[] if $\mathsf{result}_1$ can be parsed as $\sigma_1$ then
				\begin{itemize}[leftmargin=0.4cm]
					\item[] if $\validate(\sigma_1, \cP^\ell_N)=1$ then 
					\begin{itemize}[leftmargin=0.4cm]
						\item[] if $\mathsf{result}_2$ can be parsed as $\sigma_2$ then
						\begin{itemize}[leftmargin=0.4cm]
							\item[] if $\validate(\sigma_2, \cP^\ell_N)=1$ then  {\color{Dandelion} // Clause 1}
							\begin{itemize}[leftmargin=0.4cm]
								\item[] $\cLedger[\mathcal{R}_i]:=\cLedger[\mathcal{R}_i]+\$(\frac{p}{2}+ {d_F})$, $\forall i \in \{1,2\}$
								\item[] $\cLedger[\mathcal{LW}]:=\cLedger[\mathcal{LW}]+\$(e)$ 
							\end{itemize}
							\item[] else {\color{gray}// i.e., $\validate(\sigma_2, \cP^\ell_N)=0$} {\color{Dandelion} // Clause 2}
							\begin{itemize}[leftmargin=0.4cm]
								\item[] $\cLedger[\mathcal{R}_1]:=\cLedger[\mathcal{R}_1]+\$(p+\frac{3}{2}d_F)$
								\item[] $\cLedger[\mathcal{LW}]:=\cLedger[\mathcal{LW}]+\$(e+\frac{d_F}{2})$ 
							\end{itemize}
						\end{itemize}
						\item[] else {\color{gray}// i.e., $\mathsf{result}_2=\bot$} {\color{Dandelion} // Clause 3}
						\begin{itemize}[leftmargin=0.4cm]
							\item[] $\cLedger[\mathcal{R}_1]:=\cLedger[\mathcal{R}_1]+\$(p+\frac{3}{2}d_F)$
							\item[] $\cLedger[\mathcal{LW}]:=\cLedger[\mathcal{LW}]+\$(e+\frac{d_F}{2})$ 
						\end{itemize}
					\end{itemize}
					\item[] else {\color{gray}// i.e., $\validate(\sigma_1, \cP^\ell_N)=0$}
					\begin{itemize}[leftmargin=0.4cm]
						\item[] if $\mathsf{result}_2$ can be parsed as $\sigma_2$ then
						\begin{itemize}[leftmargin=0.4cm]
							\item[] if $\validate(\sigma_2, \cP^\ell_N)=1$ then {\color{Dandelion} // Clause 2}
							\begin{itemize}[leftmargin=0.4cm]
								\item[] $\cLedger[\mathcal{R}_2]:=\cLedger[\mathcal{R}_2]+\$(p+\frac{3}{2}d_F)$
								\item[] $\cLedger[\mathcal{LW}]:=\cLedger[\mathcal{LW}]+\$(e+\frac{d_F}{2})$
							\end{itemize}
							\item[] else {\color{gray}// i.e., $\validate(\sigma_2, \cP^\ell_N)=0$} {\color{Dandelion} // Clause 4}
							\begin{itemize}[leftmargin=0.4cm]
								\item[] $\cLedger[\mathcal{LW}]:=\cLedger[\mathcal{LW}]+\$(p+e+2{d_F})$ 
							\end{itemize}
						\end{itemize}
						\item[] else {\color{gray}// i.e., $\mathsf{result}_2=\bot$}  {\color{Dandelion} // Clause 5}
						\begin{itemize}[leftmargin=0.4cm]
							\item[] $\cLedger[\mathcal{R}_1]:=\cLedger[\mathcal{R}_1]+\$(p/2-r+{d_F})$
							\item[] $\cLedger[\mathcal{LW}]:=\cLedger[\mathcal{LW}]+\$(\frac{p}{2}+e +r+{d_F})$ 
						\end{itemize}
					\end{itemize}
				\end{itemize}
				\item[] else {\color{gray}// i.e., $\mathsf{result}_1=\bot$}
				\begin{itemize}[leftmargin=0.4cm]
					\item[] if $\mathsf{result}_2$ can be parsed as $\sigma_2$ then
					\begin{itemize}[leftmargin=0.4cm]
						\item[] if $\validate(\sigma_2, \cP^\ell_N)=1$ then  {\color{Dandelion} // Clause 3}
						\begin{itemize}[leftmargin=0.4cm]
							\item[] $\cLedger[\mathcal{R}_2]:=\cLedger[\mathcal{R}_2]+\$(p+\frac{3d_F}{2})$
							\item[] $\cLedger[\mathcal{LW}]:=\cLedger[\mathcal{LW}]+\$(e+\frac{d_F}{2})$ 
						\end{itemize}
						\item[] else {\color{gray}// i.e., $\validate(\sigma_2, \cP^\ell_N)=0$} {\color{Dandelion} // Clause 5}
						\begin{itemize}[leftmargin=0.4cm]
							\item[] $\cLedger[\mathcal{R}_2]:=\cLedger[\mathcal{R}_2]+\$(p/2-r+{d_F})$
							\item[] $\cLedger[\mathcal{LW}]:=\cLedger[\mathcal{LW}]+\$(\frac{p}{2}+e +r+{d_F})$ 
						\end{itemize}
					\end{itemize}
					\item[] else {\color{gray}// i.e., $\mathsf{result}_2=\bot$} {\color{Dandelion} // Clause 6}
					\begin{itemize}[leftmargin=0.4cm]
						\item[] $\cLedger[\mathcal{R}_i]:=\cLedger[\mathcal{R}_i]+\$(\frac{p}{2}-r+{d_F})$, $\forall i \in \{1,2\}$
						\item[] $\cLedger[\mathcal{LW}]:=\cLedger[\mathcal{LW}]+\$(e+2r)$ 
					\end{itemize}
				\end{itemize}
				
			\end{itemize}
			
			\vspace{-3mm}
			
			
			}
		}%
	}
	\caption{$\mathsf{Payout}$ subroutine called by Fig 9.}	\label{fig:payout}
	\vspace{-3mm}
\end{figure}

\begin{figure}[H]
	\centering
	\fbox{%
		\parbox{0.75\linewidth}{%
			\vspace{-3mm}
			\begin{center}
				{\bf {\color{Lavender}$\mathsf{Payout}^\prime$} subroutine} 
			\end{center}
			\vspace{-3mm}
			{\scriptsize
			{\color{Lavender}$\mathsf{Payout}^\prime$}$(\mathsf{result}_i, \cP^\ell_N)$:
			\vspace{-3mm}

			\begin{itemize}[leftmargin=0.4cm]
				\item[] if $\mathsf{result}_i$ can be parsed as $\sigma_i$ then
				\begin{itemize}[leftmargin=0.4cm]
					\item[] if $\validate(\sigma_i, \cP^\ell_N)=1$ then {\color{Dandelion} // Clause 7}
					\begin{itemize}
						\item[] $\cLedger[\mathcal{R}_i]:=\cLedger[\mathcal{R}_i]+\$(p+{d_F})$
						\item[] $\cLedger[\mathcal{R}_{|1-i|}]:=\cLedger[\mathcal{R}_{|1-i|}]+\$({d_F})$
						\item[] $\cLedger[\mathcal{LW}]:=\cLedger[\mathcal{LW}]+\$(\frac{e}{2})$
					\end{itemize}
					\item[] else  {\color{gray}// i.e., $\validate(\sigma_i, \cP^\ell_N)=0$} {\color{Dandelion} // Clause 8}
					\begin{itemize}
						\item[] $\cLedger[\mathcal{R}_{|1-i|}]:=\cLedger[\mathcal{R}_{|1-i|}]+\$({d_F})$
						\item[] $\cLedger[\mathcal{LW}]:=\cLedger[\mathcal{LW}]+\$(\frac{p+e}{2}+\frac{d_F}{2})$
					\end{itemize}
				\end{itemize}
				
				\item[] else  {\color{gray}// i.e., $\mathsf{result}_i=\bot$} {\color{Dandelion} // Clause 9}
				\begin{itemize}
					\item[] $\cLedger[\mathcal{R}_i]:=\cLedger[\mathcal{R}_i]+\$(\frac{p}{2}-r+{d_F})$
					\item[] $\cLedger[\mathcal{R}_{|1-i|}]:=\cLedger[\mathcal{R}_{|1-i|}]+\$({d_F})$
					\item[] $\cLedger[\mathcal{LW}]:=\cLedger[\mathcal{LW}]+\$(\frac{e}{2}+r)$
				\end{itemize}
				
			\end{itemize}
			\vspace{-3mm}
			}
		}%
	}
	\caption{$\mathsf{Payout}'$ subroutine called by Fig 9.}	\label{fig:payout1}
\end{figure}

\subsection{Basic incentive for the protocol with one single relay}\label{incentive:1-1}

When the light client protocol is joined by only one single relay node (which models that the client does not believe there are non-cooperative relays), we tune the incentive to let the only relay node to audit itself by ``asymmetrically'' pay the proved claim of truthness and the unproved claim of falseness (with an extra protocol parameter $r$), and thus forwarding the correct evaluation result of the chain predicate becomes dominating.

Remark that the similar tuning can be straightly adopted by the incentive subroutine of the protocol for two relays; actually, we already incorporate the idea of introducing the extra parameter $r$ in the incentive subroutine of the protocol for two relays. For presentation simplicity, we analyze the effect of the idea of using ``asymmetrically''  payoffs for one single relay (instead of the two colluding relays).

\begin{figure}[H]
	\centering
	\fbox{%
		\parbox{0.75\linewidth}{%
			\vspace{-3mm}
			\begin{center}
				{\bf {\color{BrickRed}$\mathsf{Incentive}$}  subroutine} for one single relay
			\end{center}
			\vspace{-3mm}
			{\scriptsize
			{\color{BrickRed}$\mathsf{Incentive}$}$(${\color{TealBlue}$\mathsf{responses}$}$, \cP^\ell_N)$:
			\vspace{-3mm}
			\begin{itemize}[leftmargin=0.4cm]
				\item[] if $|\mathsf{responses}|=1$ then 
				\begin{itemize}[leftmargin=0.4cm]
					\item[] parse $\{(\mathsf{result}_i, \mathsf{sig}_i)\}:=\mathsf{responses}$
					\item[] if $\mathsf{vrfy}(\mathsf{result}_i||\mathsf{ctr}, \mathsf{sig}_i)=1$
					\begin{itemize}[leftmargin=0.4cm]
						\item[] if $\mathsf{result}_i$ can be parsed as $\sigma_i$ then
						\begin{itemize}[leftmargin=0.4cm]
							\item[] if $\validate(\sigma_i, \cP^\ell_N)=1$ then
							\begin{itemize}[leftmargin=0.4cm]
								\item[] $\cLedger[\mathcal{R}_i]:=\cLedger[\mathcal{R}_i]+\$(p+{d_F})$
								\item[] $\cLedger[\mathcal{LW}]:=\cLedger[\mathcal{LW}]+\$e$
							\end{itemize}
							\item[] else  {\color{gray}// i.e., $\validate(\sigma_i, \cP^\ell_N)=0$} 
							\begin{itemize}[leftmargin=0.4cm]
								\item[] $\cLedger[\mathcal{LW}]:=\cLedger[\mathcal{LW}]+\$(p+e+d_F)$
							\end{itemize}
						\end{itemize}
						
						\item[] else  {\color{gray}// i.e., $\mathsf{result}_i=\bot$} 
						\begin{itemize}[leftmargin=0.4cm]
							\item[] $\cLedger[\mathcal{R}_i]:=\cLedger[\mathcal{R}_i]+\$(p-r+{d_F})$
							\item[] $\cLedger[\mathcal{LW}]:=\cLedger[\mathcal{LW}]+\$(e+ r)$
						\end{itemize}
						\item[] return
					\end{itemize}
				\end{itemize}
				
				\item[] $\cLedger[\mathcal{R}_i]:=\cLedger[\mathcal{R}_i]+\${d_F}$
				\item[] $\cLedger[\mathcal{LW}]:=\cLedger[\mathcal{LW}]+\$d_L$
	
			\end{itemize}
			\vspace{-3mm}
			}
		}%
	}
	\caption{{\color{BrickRed}$\mathsf{Incentive}$} subroutine for the protocol with (one single relay).}	\label{fig:incentive1-1}
\end{figure}

\subsection{Augmented incentive   for the protocol with one relay}\label{incentive:1-2}

Here is the augmented incentive mechanism for the protocol with one single relay. Different from the aforementioned idea of using incentive to create ``self-audition'', 
we explicitly allow extra public full nodes to audit the relay node in the system. To do so, the incentive subroutine has to wait for some ``debating'' message sent from the public, once it receives the feedback from the client about what the relay node does forward. In the relay node is cheating to claim a fake unprovable side, the public can generate a proof attesting that the relay was dishonest, thus allowing the contract punish the cheating relay. 

\begin{figure}[H]
	\centering
	\fbox{%
		
		\parbox{0.75\linewidth}{%
			\vspace{-3mm}

			\begin{center}
				{\bf {\color{BrickRed}$\mathsf{Incentive}$}  subroutine} for one single relay
			\end{center}
			\vspace{-3mm}
			
			{\scriptsize
			{\color{BrickRed}$\mathsf{Incentive}$}$(${\color{TealBlue}$\mathsf{responses}$}$, \cP^\ell_N)$:
			\vspace{-3mm}

			\begin{itemize}[leftmargin=0.4cm]

				\item[] if $|\mathsf{responses}|=1$ then 
				\begin{itemize}[leftmargin=0.4cm]
					\item[] parse $\{(\mathsf{result}_i, \mathsf{sig}_i)\}:=\mathsf{responses}$
					\item[] if $\mathsf{vrfy}(\mathsf{result}_i||\mathsf{ctr}, \mathsf{sig}_i)=1$
					\begin{itemize}[leftmargin=0.4cm]
						\item[] if $\mathsf{result}_i$ can be parsed as $\sigma_i$ then
						\begin{itemize}[leftmargin=0.4cm]
							\item[] if $\validate(\sigma_i, \cP^\ell_N)=1$ then 
							\begin{itemize}[leftmargin=0.4cm]
								\item[] $\cLedger[\mathcal{R}_i]:=\cLedger[\mathcal{R}_i]+\$(p+{d_F})$
								\item[] $\cLedger[\mathcal{LW}]:=\cLedger[\mathcal{LW}]+\$e$
							\end{itemize}
							\item[] else  {\color{gray}// i.e., $\validate(\sigma_i, \cP^\ell_N)=0$} 
							\begin{itemize}[leftmargin=0.4cm]
								\item[] $\cLedger[\mathcal{LW}]:=\cLedger[\mathcal{LW}]+\$(p+e+d_F)$
							\end{itemize}
						\end{itemize}
						
						\item[] else  {\color{gray}// i.e., $\mathsf{result}_i=\bot$} 
						\begin{itemize}[leftmargin=0.4cm]
							\item[] $T_{debate}:=T+\Delta T $
							\item[] $\mathsf{debate}:=\cP^\ell_N$
						\end{itemize}
						\item[] return
					\end{itemize}
				\end{itemize}
				\item[] $\cLedger[\mathcal{R}_i]:=\cLedger[\mathcal{R}_i]+\${d_F}$
				\item[] $\cLedger[\mathcal{LW}]:=\cLedger[\mathcal{LW}]+\$d_L$
			\end{itemize}
		
			
			\begin{itemize}[leftmargin=1.2cm]
				\item[{\bf Debate.}]  Upon receiving  $(\mathsf{debating}, \cP^\ell_N, {\sigma}_{pfn})$ from $\mathcal{PFN}$:
				\begin{itemize} [leftmargin=0.4cm]
					\item[] assert $\mathsf{debate}\ne\emptyset$ and $\mathsf{debate}=\cP^\ell_N$
					\item[] if $\validate({\sigma}_{pfn}, \cP^\ell_N)=1$ then:
					\begin{itemize} [leftmargin=0.4cm]
						\item[] $\cLedger[\mathcal{PFN}]:=\cLedger[\mathcal{PFN}]+\${d_F}$
						\item[] $\cLedger[\mathcal{LW}]:=\cLedger[\mathcal{LW}]+\$(p+e)$
						\item[]  $\mathsf{debate}:=\emptyset$
					\end{itemize}
				\end{itemize}
				\item[{\bf Timer.}] Upon $T \geq  T_{debate}$ and $\mathsf{debate}\ne\emptyset$:
				\begin{itemize} [leftmargin=0.4cm]
						\item[] $\cLedger[\mathcal{R}_i]:=\cLedger[\mathcal{R}_i]+\$({d_F}+p)$
						\item[] $\cLedger[\mathcal{LW}]:=\cLedger[\mathcal{LW}]+\$e$
						\item[]  $\mathsf{debate}:=\emptyset$
				\end{itemize}
			\end{itemize}
		}
		
		\vspace{-3mm}
		
		}%
	}
	\caption{Augmented {\color{BrickRed}$\mathsf{Incentive}$} subroutine (a single relay) with another public full node.}	\label{fig:incentive1-2}
	\vspace{-3mm}
\end{figure}

\section{Inductive increments of utilities}\label{append:game}

Let $h$ denote a history in $\Gamma_2^k$ the to represent the beginning of a reachable query, at which the $\mathcal{LW}$ moves to determine whether to abort or not, and then $chance$, $\mathcal{R}_1$, $\mathcal{R}_2$ and $\mathcal{LW}$ sequentially move.
Then the utilities of  $\mathcal{LW}$, $\mathcal{R}_1$ and $\mathcal{R}_2$ can be defined recursively as shown in Table \ref{tab:utility-a}, \ref{tab:utility-a1} and \ref{tab:utility-o}.

We make the following remarks about the recursive definition of the utilities:
(i)  when $h:=\emptyset$, set $u_{\mathcal{LW}}(h):=0$, $u_{\mathcal{R}_1}(h):=0$, and $u_{\mathcal{R}_2}(h):=0$;
(ii) $u_{\mathcal{LW}}(h) = u_{\mathcal{LW}}(h) +  d_L$ and the utility functions depict the final outcomes of the players, when $|h|=6k$ (i.e. the protocol expires);
(iii) when the $chance$ choose $a$ and the lightweight client outputs $A'$, or when the $chance$ choose $a'$ and the lightweight client outputs $A$, the light client is fooled, which means its utility increment shall subtract $v$, and the utility increments of $\mathcal{R}_1$ and $\mathcal{R}_2$ shall add $v_1:=v_1(\cP^{1}_{h},\cC)$ and $v_2:=v_2(\cP^{1}_{h},\cC)$ respectively;
(iv) when the light client outputs $O$, which means the client decides to setup its own full node and $c$ is subtracted from its utility increment to reflect the cost.
%
%
Also note that if two non-cooperative relays $\mathcal{R}_1$ and $\mathcal{R}_2$ share no common conflict of cheating the light client, one more constrain is applied to ensure   $v_1(\cP^{1}_{h}, \cC) \cdot v_2(\cP^{1}_{h},\cC)=0$.


The utility functions of game $\Gamma_1^k$ and $G_1^k$ can be inductively defined similarly.

\renewcommand{\arraystretch}{1.2}
\begin{table*}[]
	\centering
	\caption{The recursive definition of utility function: $h$ denotes a history in $\Gamma_2^k$ the to represent any history at which $LW$ moves to determine raise a query or abort the protocol ``-'' means an unreachable history. (To be continued in Table \ref{tab:utility-a1}) 
	} \label{tab:utility-a}
	\tiny
	\begin{tabular}{l|llllll} 
		\hline\hline
		\begin{tabular}[c]{@{}l@{}}Info Set \\ of $LW$\end{tabular} & \begin{tabular}[c]{@{}l@{}}Histories \\ of $LW$\end{tabular} & \begin{tabular}[c]{@{}l@{}}Actions \\ of $LW$~\end{tabular} & Utility of $LW$           & Utility of $R_1$     & Utility of $R_2$       \ignore{&}           \\

		\hline\hline
		\multirow{8}{*}{$I^1_{LW}$ }                                         
		& \multirow{4}{*}{$hQatt$ }                                        & $TA$                                                            & $u_{LW}(h)+d_L-p$               & $u_{R_1}(h)+p/2+d_F$            & $u_{R_2}(h)+p/2+d_F$          \ignore{& Clause 1}   \\
		&                                                                 & $LA$                                                            & $u_{LW}(h)+d_L-p-e/2$           & -                           & -                           \ignore{& Clause  7}   \\
		&                                                                 & $RA$                                                            & $u_{LW}(h)+d_L-p-e/2$           & -                           & -                           \ignore{& Clause  7}   \\
		&                                                                 & $XA$                                                            & 
		$u_{LW}(h)+d_L-p-e$         & -                           & -                          \ignore{& Clause 10}  \\ 
		\cline{2-7}
		& \multirow{4}{*}{$hQa'ff$ }                                       & $TA$                                                           &
		$u_{LW}(h)+d_L-v+2d_F$        & $u_{R_1}(h)+v_1$    & $u_{R_2}(h)+v_2$  \ignore{& Clause  4}  \\
		&                                                                 & $LA$                                                            & 
		$u_{LW}(h)+d_L-v-p/2-e/2+d_F/2$  & -                           & -                           \ignore{& Clause  8}   \\
		&                                                                 & $RA$                                                            & 
		$u_{LW}(h)+d_L-v-p/2-e/2+d_F/2$  & -                           & -                           \ignore{& Clause  8}   \\
		&                                                                 & $XA$                                                            & 
		$u_{LW}(h)+d_L-v-p-e$       & -                           & -                           \ignore{& Clause 10}  \\ 
		\hline\hline
		\multirow{8}{*}{$I^2_{LW}$ }                                         
		& \multirow{4}{*}{$hQatf$ }                                        & $TA$                                                            & 
		$u_{LW}(h)+d_L-p+d_F $          & $u_{R_1}(h)+p+3d_F/2$     & $u_{R_2}(h)$        \ignore{& Clause  2, 3}   \\
		&                                                                 & $LA$                                                            & $u_{LW}(h)+d_L-p-e/2$             & -                           & -                           \ignore{& Clause  7}   \\
		&                                                                 & $RA$                                                            & $u_{LW}(h)+d_L-p-e/2+r$             & -                           & -                           \ignore{& Clause  9}   \\
		&                                                                 & $XA$                                                            & 
		$u_{LW}(h)+d_L-p-e$           & -                           & -                           \ignore{& Clause 10}  \\ 
		\cline{2-7}
		& \multirow{4}{*}{$hQa'ft$ }                                       & $TA$                                                            & $u_{LW}(h)+d_L-v-p/2+r+d_F $     & $u_{R_1}(h)+v_1$        & $u_{R_2}(h)+v_2+p/2-r+d_F $        & \ignore{Clause 5}   \\
		&                                                                 & $LA$                                                            & $u_{LW}(h)+d_L-v-p/2-e/2+d_F/2$      & -                           & -                           & \ignore{Clause  8}   \\
		&                                                                 & $RA$                                                            & $u_{LW}(h)+d_L-v-p-e/2+r$             & -                           & -                           & \ignore{Clause  9}   \\
		&                                                                 & $XA$                                                            & 
		$u_{LW}(h)+d_L-v-p-e$           & -                           & -                           \ignore{& Clause 10}  \\ 
		\hline\hline
		\multirow{4}{*}{$I^3_{LW}$ }                                         
		& \multirow{2}{*}{$hQatx$ }                                        & $TA$                                                            & $u_{LW}(h)+d_L-p-e/2$             & $u_{R_1}(h)+p+d_F $            & $u_{R_2}(h)+d_F $              & \ignore{Clause  7}   \\
		&                                                                 & $XA$                                                            & 
		$u_{LW}(h)+d_L-p-e$           & -                           & -                          \ignore{& Clause 10}  \\ 
		\cline{2-7}
		& \multirow{2}{*}{$hQa'fx$ }                                       & $TA$                                                            & $u_{LW}(h)+d_L-v-p/2-e/2+d_F/2$      & $u_{R_1}(h)+v_1$        & $u_{R_2}(h)+v_2+d_F $              & \ignore{Clause  8}   \\
		&                                                                 & $XA$                                                            & 
		$u_{LW}(h)+d_L-v-p-e$           & -                           & -                           \ignore{& Clause 10}  \\ 
		\hline\hline
		\multirow{8}{*}{$I^4_{LW}$ }                                         
		& \multirow{4}{*}{$hQaft$ }                                        & $TA$                                                            & $u_{LW}(h)+d_L-p+d_F/2$          & $u_{R_1}(h)$        & $u_{R_2}(h)+p+3d_F/2$     & \ignore{Clause  2, 3}   \\
		&                                                                 & $LA$                                                            & $u_{LW}(h)+d_L-p-e/2+r$             & -                           & -                           & \ignore{Clause  9}   \\
		&                                                                 & $RA$                                                            & $u_{LW}(h)+d_L-p-e/2$             & -                           & -                           & \ignore{Clause  7}   \\
		&                                                                 & $XA$                                                            & 
		$u_{LW}(h)+d_L-p-e$           & -                           & -                           \ignore{& Clause 10}  \\ 
		\cline{2-7}
		& \multirow{4}{*}{$hQa'tf$ }                                       & $TA$                                                            & $u_{LW}(h)+d_L-v-p/2+r+d_F $     & $u_{R_1}(h)+v_1+p/2-r+d_F $        & $u_{R_2}(h)+v_2$        \ignore{& Clause 5}   \\
		&                                                                 & $LA$                                                            & $u_{LW}(h)+d_L-v-p-e/2+r$             & -                           & -                           \ignore{& Clause  9}   \\
		&                                                                 & $RA$                                                            & 
		$u_{LW}(h)+d_L-v-p/2-e/2+d_F/2$      & -                           & -                           \ignore{& Clause  8}   \\
		&                                                                 & $XA$                                                            & 
		$u_{LW}(h)+d_L-v-p-e$           & -                           & -                           \ignore{& Clause 10}  \\ 
		\hline\hline
		\multirow{8}{*}{$I^5_{LW}$ }                                         
		& \multirow{4}{*}{$hQaff$ }                                        & $TA$                                                            & 
		$u_{LW}(h)+d_L-p+2r$         & $u_{R_1}(h)+p/2-r+d_F $  & $u_{R_2}(h)+p/2-r+d_F $  \ignore{& Clause  6}   \\
		&                                                                 & $LA$                                                            & 
		$u_{LW}(h)+d_L-p-e/2+r$         & -                           & -                           \ignore{& Clause  9}   \\
		&                                                                 & $RA$                                                            & 
		$u_{LW}(h)+d_L-p-e/2+r$         & -                           & -                           \ignore{& Clause  9}   \\
		&                                                                 & $XA$                                                            & 
		$u_{LW}(h)+d_L-p-e$       & -                           & -                           \ignore{& Clause 10}  \\ 
		\cline{2-7}
		& \multirow{4}{*}{$hQa'tt$ }                                       & $TA$                                                            & 
		$u_{LW}(h)+d_L-v -p+2r$           & $u_{R_1}(h)+v_1+p/2-r+d_F $        & $u_{R_2}(h)+v_2+p/2-r+d_F $        \ignore{& Clause  6}   \\
		&                                                                 & $LA$                                                            & 
		$u_{LW}(h)+d_L-v-p-e/2+r$           & -                           & -                           \ignore{& Clause  9}   \\
		&                                                                 & $RA$                                                            & 
		$u_{LW}(h)+d_L-v-p-e/2+r$           & -                           & -                           \ignore{& Clause  9}   \\
		&                                                                 & $XA$                                                            & 
		$u_{LW}(h)+d_L-v-p-e$         & -                           & -                           \ignore{& Clause 10}  \\ 
		\hline\hline
		\multirow{4}{*}{$I^6_{LW}$ }                                         
		& \multirow{2}{*}{$hQafx$ }                                        & $TA$                                                            & $u_{LW}(h)+d_L-p-e/2+r$             & $u_{R_1}(h)+p/2-r+d_F $        & $u_{R_2}(h)+d_F $              \ignore{& Clause  9}   \\
		&                                                                 & $XA$                                                            & 
		$u_{LW}(h)+d_L-p-e$           & -                           & -                           \ignore{& Clause 10}  \\ 
		\cline{2-7}
		& \multirow{2}{*}{$hQa'tx$ }                                       & $TA$                                                            & $u_{LW}(h)+d_L-v-p-e/2+r$              & $u_{R_1}(h)+v_1+p/2-r+d_F $         & $u_{R_2}(h)+v_2+d_F $              \ignore{& Clause  9}   \\
		&                                                                 & $XA$                                                            & 
		$u_{LW}(h)+d_L-v-p-e$           & -                           & -                           \ignore{& Clause 10}  \\ 
		\hline\hline
		\multirow{4}{*}{$I^7_{LW}$ }                                         
		& \multirow{2}{*}{$hQaxt$ }                                        & $TA$                                                            & $u_{LW}(h)+d_L-p-e/2$             & $u_{R_1}(h)+d_F $              & $u_{R_2}(h)+p+d_F $            \ignore{& Clause  7}   \\
		&                                                                 & $XA$                                                            & 
		$u_{LW}(h)+d_L-p-e$           & -                           & -                           \ignore{& Clause 10}  \\ 
		\cline{2-7}
		& \multirow{2}{*}{$hQa'xf$ }                                       & $TA$                                                            & $u_{LW}(h)+d_L-v-p/2-e/2+d_F/2$      & $u_{R_1}(h)+v_1+d_F $              & $u_{R_2}(h)+v_2$       \ignore{& Clause  8}   \\
		&                                                                 & $XA$                                                            & 
		$u_{LW}(h)+d_L-v-p-e$           & -                           & -                           \ignore{& Clause 10}  \\ 
		\hline\hline
		\multirow{4}{*}{$I^8_{LW}$ }                                         
		& \multirow{2}{*}{$hQaxf$ }                                        & $TA$                                                            & 
		$u_{LW}(h)+d_L-p-e/2+r$             & $u_{R_1}(h)+d_F $              & $u_{R_2}(h)+p/2-r+d_F $         \ignore{& Clause  9}   \\
		&                                                                 & $XA$                                                            & 
		$u_{LW}(h)+d_L-p-e$           & -                           & -                            \ignore{&  Clause  10}  \\ 
		\cline{2-7}
		& \multirow{2}{*}{$hQa'xt$ }                                       & $TA$                                                            & $u_{LW}(h)+d_L-v-p-e/2+r$             & $u_{R_1}(h)+v_1+d_F $              & $u_{R_2}(h)+v_2+p/2-r+d_F $         \ignore{& Clause  9}   \\
		&                                                                 & $XA$                                                            & 
		$u_{LW}(h)+d_L-v-p-e$           & -                           & -                            \ignore{& Clause  10}  \\ 
		\hline\hline
		\multirow{4}{*}{$I^9_{LW}$ }                                         
		& \multirow{2}{*}{$hQaxx$ }                                        & $TA$                                                            & 
		$u_{LW}(h)+d_L-p-e$             & $u_{R_1}(h)+d_F $              & $u_{R_2}(h)+d_F $         \ignore{& Clause  9}   \\
		&                                                                 & $XA$                                                            & 
		$u_{LW}(h)+d_L-p-e$           & -                           & -                            \ignore{&  Clause  10}  \\ 
		\cline{2-7}
		& \multirow{2}{*}{$hQa'xx$ }                                       & $TA$                                                            & $u_{LW}(h)+d_L-v-p-e$             & $u_{R_1}(h)+v_1+d_F $              & $u_{R_2}(h)+v_2+d_F $         \ignore{& Clause  9}   \\
		&                                                                 & $XA$                                                            & 
		$u_{LW}(h)+d_L-v-p-e$           & -                           & -                            \ignore{&  Clause  10} \\ 
		\hline\hline

		-   & $h$ & $B$     & $u_{LW}(h)$           & -              & -              &   \\
		\hline\hline
		
	\end{tabular}
	
\end{table*}

\renewcommand{\arraystretch}{1.2}
\begin{table*}[]
	\centering
	\caption{Continuation of Table.\ref{tab:utility-a}: The recursive definition of utility function: $h$ denotes a history in $\Gamma_2^k$ the to represent any history at which $LW$ moves to determine raise a query or abort the protocol ``-'' means an unreachable history. } \label{tab:utility-a1}
	\tiny
	\begin{tabular}{l|llllll} 
		\hline\hline
		\begin{tabular}[c]{@{}l@{}}Info Set \\ of $LW$\end{tabular} & \begin{tabular}[c]{@{}l@{}}Histories \\ of $LW$\end{tabular} & \begin{tabular}[c]{@{}l@{}}Actions \\ of $LW$~\end{tabular} & Utility of $LW$           & Utility of $R_1$     & Utility of $R_2$     \ignore{&}            \\

		\hline\hline
		\multirow{8}{*}{$I^1_{LW}$ }                                         
		& \multirow{4}{*}{$hQatt$ }                                        & $TA'$                                                            & $u_{LW}(h)+d_L-v-p$               & $u_{R_1}(h)+v_1+p/2+d_F $            & $u_{R_2}(h)+v_2+p/2+d_F $          \ignore{& Clause 1}   \\
		&                                                                 & $LA'$                                                            & $u_{LW}(h)+d_L-v-p-e/2$           & -                           & -                           \ignore{& Clause  7}   \\
		&                                                                 & $RA'$                                                            & $u_{LW}(h)+d_L-v-p-e/2$           & -                           & -                           \ignore{& Clause  7}   \\
		&                                                                 & $XA'$                                                            & 
		$u_{LW}(h)+d_L-v-p-e$         & -                           & -                           \ignore{& Clause 10}  \\ 
		\cline{2-7}
		& \multirow{4}{*}{$hQa'ff$ }                                       & $TA'$                                                           &
		$u_{LW}(h)+d_L+2d_F $        & $u_{R_1}(h)$    & $u_{R_2}(h)$  \ignore{& Clause  4}   \\
		&                                                                 & $LA'$                                                            & 
		$u_{LW}(h)+d_L-p/2-e/2+d_F/2$  & -                           & -                           \ignore{& Clause  8}   \\
		&                                                                 & $RA'$                                                            & 
		$u_{LW}(h)+d_L-p/2-e/2+d_F/2$  & -                           & -                           \ignore{& Clause  8}   \\
		&                                                                 & $XA'$                                                            & 
		$u_{LW}(h)+d_L-p-e$       & -                           & -                           \ignore{& Clause 10}  \\ 
		\hline\hline
		\multirow{8}{*}{$I^2_{LW}$ }                                         
		& \multirow{4}{*}{$hQatf$ }                                        & $TA'$                                                            & 
		$u_{LW}(h)+d_L-v-p+d_F/2$          & $u_{R_1}(h)+v_1+p+3d_F/2$     & $u_{R_2}(h)+v_2$        \ignore{& Clause  2, 3}   \\
		&                                                                 & $LA'$                                                            & $u_{LW}(h)+d_L-v-p-e/2$             & -                           & -                            \ignore{& Clause  7}   \\
		&                                                                 & $RA'$                                                            & $u_{LW}(h)+d_L-v-p-e/2+r$             & -                           & -                            \ignore{& Clause  9}   \\
		&                                                                 & $XA'$                                                            & 
		$u_{LW}(h)+d_L-v-p-e$           & -                           & -                           \ignore{& Clause 10}  \\ 
		\cline{2-7}
		& \multirow{4}{*}{$hQa'ft$ }                                       & $TA'$                                                            & 
		$u_{LW}(h)+d_L-p/2+r+d_F $     & $u_{R_1}(h)$        & $u_{R_2}(h)+p/2-r+d_F $        \ignore{& Clause 5}   \\
		&                                                                 & $LA'$                                                            & $u_{LW}(h)+d_L-p/2-e/2+d_F/2$      & -                           & -                          \ignore{ & Clause  8}   \\
		&                                                                 & $RA'$                                                            & $u_{LW}(h)+d_L-p-e/2+r$             & -                           & -                           \ignore{& Clause  9}   \\
		&                                                                 & $XA'$                                                            & 
		$u_{LW}(h)+d_L-p-e$           & -                           & -                           \ignore{& Clause 10}  \\ 
		\hline\hline
		\multirow{4}{*}{$I^3_{LW}$ }                                         
		& \multirow{2}{*}{$hQatx$ }                                        & $TA'$                                                            & $u_{LW}(h)+d_L-v-p-e/2$             & $u_{R_1}(h)+v_1+p+d_F $            & $u_{R_2}(h)+v_2+d_F $              \ignore{& Clause  7}   \\
		&                                                                 & $XA'$                                                            & $u_{LW}(h)+d_L-v-p-e$           & -                           & -                           \ignore{& Clause 10}  \\ 
		\cline{2-7}
		& \multirow{2}{*}{$hQa'fx$ }                                       & $TA'$                                                            & $u_{LW}(h)+d_L-p/2-e/2+d_F/2$      & $u_{R_1}(h)$        & $u_{R_2}(h)+d_F $              \ignore{& Clause  8}   \\
		&                                                                 & $XA'$                                                            & 
		$u_{LW}(h)+d_L-p-e$           & -                           & -                           \ignore{& Clause 10}  \\ 
		\hline\hline
		\multirow{8}{*}{$I^4_{LW}$ }                                         
		& \multirow{4}{*}{$hQaft$ }                                        & $TA'$                                                            & $u_{LW}(h)+d_L-v-p+d_F/2$          & $u_{R_1}(h)+v_1$        & $u_{R_2}(h)+v_2+p+3d_F/2$     \ignore{& Clause  2, 3}   \\
		&                                                                 & $LA'$                                                            & $u_{LW}(h)+d_L-v-p-e/2+r$             & -                           & -                           \ignore{& Clause  9}   \\
		&                                                                 & $RA'$                                                            & $u_{LW}(h)+d_L-v-p-e/2$             & -                           & -                           \ignore{& Clause  7}   \\
		&                                                                 & $XA'$                                                            & $u_{LW}(h)+d_L-v-p-e$           & -                           & -                           \ignore{& Clause 10}  \\ 
		\cline{2-7}
		& \multirow{4}{*}{$hQa'tf$ }                                       & $TA'$                                                            & $u_{LW}(h)+p/2+r+d_F $     & $u_{R_1}(h)+d_L+p/2-r+d_F $        & $u_{R_2}(h)$        \ignore{& Clause 5}   \\
		&                                                                 & $LA'$                                                            & $u_{LW}(h)+d_L-p-e/2+r$             & -                           & -                           \ignore{& Clause  9}   \\
		&                                                                 & $RA'$                                                            & 
		$u_{LW}(h)+d_L-p/2-e/2+d_F/2$      & -                           & -                           \ignore{& Clause  8}   \\
		&                                                                 & $XA'$                                                            & 
		$u_{LW}(h)+d_L-p-e$           & -                           & -                           \ignore{& Clause 10}  \\ 
		\hline\hline
		\multirow{8}{*}{$I^5_{LW}$ }                                         
		& \multirow{4}{*}{$hQaff$ }                                        & $TA'$                                                            & 
		$u_{LW}(h)+d_L-v-p+2r$         & $u_{R_1}(h)+v_1+p/2-r+d_F $  & $u_{R_2}(h)+v_2+p/2-r+d_F $  \ignore{& Clause  6}   \\
		&                                                                 & $LA'$                                                            & 
		$u_{LW}(h)+d_L-v-p-e/2+r$         & -                           & -                           \ignore{& Clause  9}   \\
		&                                                                 & $RA'$                                                            & 
		$u_{LW}(h)+d_L-v-p-e/2+r$         & -                           & -                           \ignore{& Clause  9}   \\
		&                                                                 & $XA'$                                                            & 
		$u_{LW}(h)+d_L-v-p-e$       & -                           & -                           \ignore{& Clause 10}  \\ 
		\cline{2-7}
		& \multirow{4}{*}{$hQa'tt$ }                                       & $TA'$                                                            & 
		$u_{LW}(h)+d_L -p+2r$           & $u_{R_1}(h)+p/2-r+d_F $        & $u_{R_2}(h)+p/2-r+d_F $        \ignore{& Clause  6}   \\
		&                                                                 & $LA'$                                                            & 
		$u_{LW}(h)+d_L-p-e/2+r$           & -                           & -                           \ignore{& Clause  9}   \\
		&                                                                 & $RA'$                                                            & 
		$u_{LW}(h)+d_L-p-e/2+r$           & -                           & -                           \ignore{& Clause  9}   \\
		&                                                                 & $XA'$                                                            & 
		$u_{LW}(h)+d_L-p-e$         & -                           & -                           \ignore{& Clause 10}  \\ 
		\hline\hline
		\multirow{4}{*}{$I^6_{LW}$ }                                         
		& \multirow{2}{*}{$hQafx$ }                                        & $TA'$                                                            & $u_{LW}(h)+d_L-v-p-e/2+r$             & $u_{R_1}(h)+v_1+p/2-r+d_F $        & $u_{R_2}(h)+v_2+d_F $              \ignore{& Clause  9}   \\
		&                                                                 & $XA'$                                                            & 
		$u_{LW}(h)+d_L-v-p-e$           & -                           & -                           \ignore{& Clause 10}  \\ 
		\cline{2-7}
		& \multirow{2}{*}{$hQa'tx$ }                                       & $TA'$                                                            & $u_{LW}(h)+d_L-p-e/2+r$              & $u_{R_1}(h)+p/2-r+d_F $         & $u_{R_2}(h)+d_F $              \ignore{& Clause  9}   \\
		&                                                                 & $XA'$                                                            & 
		$u_{LW}(h)+d_L-p-e$           & -                           & -                           \ignore{& Clause 10}  \\ 
		\hline\hline
		\multirow{4}{*}{$I^7_{LW}$ }                                         
		& \multirow{2}{*}{$hQaxt$ }                                        & $TA'$                                                            & $u_{LW}(h)+d_L-v-p-e/2$             & $u_{R_1}(h)+v_1+d_F $              & $u_{R_2}(h)+v_2+p+d_F $            \ignore{& Clause  7}   \\
		&                                                                 & $XA'$                                                            & 
		$u_{LW}(h)+d_L-v-p-e$           & -                           & -                           \ignore{& Clause 10}  \\ 
		\cline{2-7}
		& \multirow{2}{*}{$hQa'xf$ }                                       & $TA'$                                                            & $u_{LW}(h)+d_L-p/2-e/2+d_F/2$      & $u_{R_1}(h)+d_F $              & $u_{R_2}(h)$        \ignore{& Clause  8}   \\
		&                                                                 & $XA'$                                                            & 
		$u_{LW}(h)+d_L-p-e$           & -                           & -                           \ignore{& Clause 10}  \\ 
		\hline\hline
		\multirow{4}{*}{$I^8_{LW}$ }                                         
		& \multirow{2}{*}{$hQaxf$ }                                        & $TA'$                                                            & 
		$u_{LW}(h)+d_L-v-p-e/2+r$             & $u_{R_1}(h)+v_1+d_F $              & $u_{R_2}(h)+v_2+p/2-r+d_F $        \ignore{& Clause  9}   \\
		&                                                                 & $XA'$                                                            & 
		$u_{LW}(h)+d_L-v-p-e$           & -                           & -                           \ignore{& Clause 10}  \\ 
		\cline{2-7}
		& \multirow{2}{*}{$hQa'xt$ }                                       & $TA'$                                                            & $u_{LW}(h)+d_L-p-e/2+r$             & $u_{R_1}(h)+d_F $              & $u_{R_2}(h)+p/2-r+d_F $        \ignore{& Clause  9}   \\
		&                                                                 & $XA'$                                                            & 
		$u_{LW}(h)+d_L-e$           & -                           & -                           \ignore{& Clause 10}  \\ 
		\hline\hline
		\multirow{2}{*}{$I^9_{LW}$ }                                         
		& $hQaxx$                                                          & $TA'$                                                            & $u_{LW}(h)+d_L-v-p-e$           & $u_{R_1}(h)+v_1$              & $u_{R_2}(h)+v_2$              \ignore{& Clause 10}  \\ 
		& $hQaxx$                                                          & $XA'$                                                            & $u_{LW}(h)+d_L-v-p-e$           & $u_{R_1}(h)+v_1$              & $u_{R_2}(h)+v_2$              \ignore{& Clause 10}  \\
		\cline{2-7}
		& $hQa'xx$                                                         & $TA'$                                                            & 
		$u_{LW}(h)+d_L-p-e$           & $u_{R_1}(h)$              & $u_{R_2}(h)$              \ignore{& Clause 10}  \\
				& $hQa'xx$                                                         & $XA'$                                                            & 
		$u_{LW}(h)+d_L-p-e$           & $u_{R_1}(h)$              & $u_{R_2}(h)$              \ignore{& Clause 10}  \\
		\hline\hline

		-   & $h$ & $B$     & $u_{LW}(h)$           & -              & -         &      \\
		\hline\hline
		
	\end{tabular}
	
\end{table*}


\renewcommand{\arraystretch}{1.2}
\begin{table*}[]
	\centering
	\caption{Continuation of Table.\ref{tab:utility-a}: The recursive definition of utility function: $h$ denotes a history in $\Gamma_2^k$ the to represent any history at which $LW$ moves to determine raise a query or abort the protocol ``-'' means an unreachable history.} \label{tab:utility-o}
	\tiny
	\begin{tabular}{l|llllll} 
		
		\hline\hline
		\begin{tabular}[c]{@{}l@{}}Info Set \\ of $LW$\end{tabular} & \begin{tabular}[c]{@{}l@{}}Histories \\ of $LW$\end{tabular} & \begin{tabular}[c]{@{}l@{}}Actions \\ of $LW$~\end{tabular} & Utility of $LW$           & Utility of $R_1$     & Utility of $R_2$     &            \\

		\hline\hline
		\multirow{8}{*}{$I^1_{LW}$ }                                         
		& \multirow{4}{*}{$hQatt$ }                                        & $TO$                                                            & $u_{LW}(h)-c-p$               & $u_{R_1}(h)+p/2+d_F $            & $u_{R_2}(h)+p/2+d_F $          \ignore{& Clause 1}    \\
		&                                                                 & $LO$                                                            & $u_{LW}(h)-c-p-e$           & -                           & -                           \ignore{& Clause  7}    \\
		&                                                                 & $RO$                                                            & $u_{LW}(h)-c-p-e$           & -                           & -                           \ignore{& Clause  7}    \\
		&                                                                 & $XO$                                                            & 
		$u_{LW}(h)-c-p-e+\epsilon$         & -                           & -                           \ignore{& Clause 10}   \\ 
		\cline{2-7}
		& \multirow{4}{*}{$hQa'ff$ }                                       & $TO$                                                           &
		$u_{LW}(h)-c+2d_F $        & $u_{R_1}(h)$    & $u_{R_2}(h)$  \ignore{& Clause  4}    \\
		&                                                                 & $LO$                                                            & 
		$u_{LW}(h)-c-p/2-e/2+d_F/2$  & -                           & -                           \ignore{& Clause  8}    \\
		&                                                                 & $RO$                                                            & 
		$u_{LW}(h)-c-p/2-e/2+d_F/2$  & -                           & -                           \ignore{& Clause  8}    \\
		&                                                                 & $XO$                                                            & 
		$u_{LW}(h)-c-p-e+\epsilon$       & -                           & -                           \ignore{& Clause 10}   \\ 
		\hline\hline
		\multirow{8}{*}{$I^2_{LW}$ }                                         
		& \multirow{4}{*}{$hQatf$ }                                        & $TO$                                                            & 
		$u_{LW}(h)-c-p+d_F/2$          & $u_{R_1}(h)+p+3d_F/2$     & $u_{R_2}(h)$        \ignore{& Clause  2, 3}    \\
		&                                                                 & $LO$                                                            & $u_{LW}(h)-c-p-e/2$             & -                           & -                           \ignore{& Clause  7}    \\
		&                                                                 & $RO$                                                            & $u_{LW}(h)-c-p-e/2+r$             & -                           & -                           \ignore{& Clause  9}    \\
		&                                                                 & $XO$                                                            & 
		$u_{LW}(h)-c-p-e$           & -                           & -                           \ignore{& Clause 10}   \\ 
		\cline{2-7}
		& \multirow{4}{*}{$hQa'ft$ }                                       & $TO$                                                            & $u_{LW}(h)-c-p/2+r+d_F $     & $u_{R_1}(h)$        & $u_{R_2}(h)+p/2-r+d_F $        \ignore{& Clause 5}    \\
		&                                                                 & $LO$                                                            & $u_{LW}(h)-c-p/2-e/2+d_F/2$      & -                           & -                           \ignore{& Clause  8}    \\
		&                                                                 & $RO$                                                            & $u_{LW}(h)-c-p-e/2+r$             & -                           & -                           \ignore{& Clause  9}    \\
		&                                                                 & $XO$                                                            & 
		$u_{LW}(h)-c-p-e$           & -                           & -                           \ignore{& Clause 10}   \\ 
		\hline\hline
		\multirow{4}{*}{$I^3_{LW}$ }                                         
		& \multirow{2}{*}{$hQatx$ }                                        & $TO$                                                            & $u_{LW}(h)-c-p-e/2$             & $u_{R_1}(h)+p+d_F $            & $u_{R_2}(h)+d_F $              \ignore{& Clause  7}    \\
		&                                                                 & $XO$                                                            & $u_{LW}(h)-c-p-e$           & -                           & -                           \ignore{& Clause 10}   \\ 
		\cline{2-7}
		& \multirow{2}{*}{$hQa'fx$ }                                       & $TO$                                                            & 
		$u_{LW}(h)-c-p/2-e/2+d_F/2$      & $u_{R_1}(h)$        & $u_{R_2}(h)+d_F $              \ignore{& Clause  8}    \\
		&                                                                 & $XO$                                                            & 
		$u_{LW}(h)-c-p-e$           & -                           & -                           \ignore{& Clause 10}   \\ 
		\hline\hline
		\multirow{8}{*}{$I^4_{LW}$ }                                         
		& \multirow{4}{*}{$hQaft$ }                                        & $TO$                                                            & $u_{LW}(h)-c-p+d_F/2$          & $u_{R_1}(h)$        & $u_{R_2}(h)+p+3d_F/2$     \ignore{& Clause  2, 3}    \\
		&                                                                 & $LO$                                                            & $u_{LW}(h)-c-p-e/2+r$             & -                           & -                           \ignore{& Clause  9}    \\
		&                                                                 & $RO$                                                            & $u_{LW}(h)-c-p-e/2$             & -                           & -                           \ignore{& Clause  7}    \\
		&                                                                 & $XO$                                                            & $u_{LW}(h)-c-p-e$           & -                           & -                           \ignore{& Clause 10}   \\ 
		\cline{2-7}
		& \multirow{4}{*}{$hQa'tf$ }                                       & $TO$                                                            & 
		$u_{LW}(h)-c+p/2+r+d_F $     & $u_{R_1}(h)+p/2-r+d_F $        & $u_{R_2}(h)$        \ignore{& Clause 5}    \\
		&                                                                 & $LO$                                                            & $u_{LW}(h)-c-p-e/2+r$             & -                           & -                           \ignore{& Clause  9}    \\
		&                                                                 & $RO$                                                            & 
		$u_{LW}(h)-c-p/2-e/2+d_F/2$      & -                           & -                           \ignore{& Clause  8}    \\
		&                                                                 & $XO$                                                            & 
		$u_{LW}(h)-c-p-e$           & -                           & -                           \ignore{& Clause 10}   \\ 
		\hline\hline
		\multirow{8}{*}{$I^5_{LW}$ }                                         
		& \multirow{4}{*}{$hQaff$ }                                        & $TO$                                                            & 
		$u_{LW}(h)-c-p+2r$         & $u_{R_1}(h)+p/2-r+d_F $  & $u_{R_2}(h)+p/2-r+d_F $  \ignore{& Clause  6}    \\
		&                                                                 & $LO$                                                            & 
		$u_{LW}(h)-c-p-e/2+r$         & -                           & -                           \ignore{& Clause  9}    \\
		&                                                                 & $RO$                                                            & 
		$u_{LW}(h)-c-p-e/2+r$         & -                           & -                           \ignore{& Clause  9}    \\
		&                                                                 & $XO$                                                            & 
		$u_{LW}(h)-c-p-e$       & -                           & -                           \ignore{& Clause 10}   \\ 
		\cline{2-7}
		& \multirow{4}{*}{$hQa'tt$ }                                       & $TO$                                                            & 
		$u_{LW}(h)-c -p+2r$           & $u_{R_1}(h)+p/2-r+d_F $        & $u_{R_2}(h)+p/2-r+d_F $        \ignore{& Clause  6}    \\
		&                                                                 & $LO$                                                            & 
		$u_{LW}(h)-c-p-e/2+r$           & -                           & -                           \ignore{& Clause  9}    \\
		&                                                                 & $RO$                                                            & 
		$u_{LW}(h)-c-p-e/2+r$           & -                           & -                           \ignore{& Clause  9}    \\
		&                                                                 & $XO$                                                            & 
		$u_{LW}(h)-c-p-e$         & -                           & -                           \ignore{& Clause 10}   \\ 
		\hline\hline
		\multirow{4}{*}{$I^6_{LW}$ }                                         
		& \multirow{2}{*}{$hQafx$ }                                        & $TO$                                                            & $u_{LW}(h)-c-p-e/2+r$             & $u_{R_1}(h)+p/2-r+d_F $        & $u_{R_2}(h)+d_F $              \ignore{& Clause  9}    \\
		&                                                                 & $XO$                                                            & 
		$u_{LW}(h)-c-p-e$           & -                           & -                           \ignore{& Clause 10}   \\ 
		\cline{2-7}
		& \multirow{2}{*}{$hQa'tx$ }                                       & $TO$                                                            & $u_{LW}(h)-c-p-e/2+r$              & $u_{R_1}(h)+p/2-r+d_F $         & $u_{R_2}(h)+d_F $             \ignore{ & Clause  9}    \\
		&                                                                 & $XO$                                                            & 
		$u_{LW}(h)-c-p-e$           & -                           & -                           \ignore{& Clause 10}   \\ 
		\hline\hline
		\multirow{4}{*}{$I^7_{LW}$ }                                         
		& \multirow{2}{*}{$hQaxt$ }                                        & $TO$                                                            & $u_{LW}(h)-c-p-e/2$             & $u_{R_1}(h)+d_F $              & $u_{R_2}(h)+p+d_F $            \ignore{& Clause  7}    \\
		&                                                                 & $XO$                                                            & 
		$u_{LW}(h)-c-p-e$           & -                           & -                           \ignore{& Clause 10}   \\ 
		\cline{2-7}
		& \multirow{2}{*}{$hQa'xf$ }                                       & $TO$                                                            & $u_{LW}(h)-c-p/2-e/2+d_F/2$      & $u_{R_1}(h)+d_F $              & $u_{R_2}(h)$        \ignore{& Clause  8}   \\
		&                                                                 & $XO$                                                            & 
		$u_{LW}(h)-c-p-e$           & -                           & -                           \ignore{& Clause 10}   \\ 
		\hline\hline
		\multirow{4}{*}{$I^8_{LW}$ }                                         
		& \multirow{2}{*}{$hQaxf$ }                                        & $TO$                                                            & 
		$u_{LW}(h)-c-p-e/2+r$             & $u_{R_1}(h)+d_F $              & $u_{R_2}(h)+p/2-r+d_F $        \ignore{& Clause  9}   \\
		&                                                                 & $XO$                                                            & 
		$u_{LW}(h)-c-p-e$           & -                           & -                           \ignore{& Clause 10}  \\ 
		\cline{2-7}
		& \multirow{2}{*}{$hQa'xt$ }                                       & $TO$                                                            & $u_{LW}(h)-c-p-e/2+r$             & $u_{R_1}(h)+d_F $              & $u_{R_2}(h)+p/2-r+d_F $        \ignore{& Clause  9}   \\
		&                                                                 & $XO$                                                            & 
		$u_{LW}(h)-c-e$           & -                           & -                           \ignore{& Clause 10}  \\ 
		\hline\hline
		\multirow{2}{*}{$I^9_{LW}$ }                                         
		& $hQaxx$                                                          & $TO$                                                            & $u_{LW}(h)-c-p-e$           & $u_{R_1}(h)$              & $u_{R_2}(h)$              \ignore{& Clause 10}  \\ 
		& $hQaxx$                                                          & $XO$                                                            & $u_{LW}(h)-c-p-e$           & $u_{R_1}(h)$              & $u_{R_2}(h)$              \ignore{& Clause 10}  \\ 
		\cline{2-7}
		& $hQa'xx$                                                         & $TO$                                                            & 
		$u_{LW}(h)-c-p-e$           & $u_{R_1}(h)$              & $u_{R_2}(h)$              \ignore{& Clause 10}  \\
		& $hQa'xx$                                                         & $XO$                                                            & 
		$u_{LW}(h)-c-p-e$           & $u_{R_1}(h)$              & $u_{R_2}(h)$              \ignore{& Clause 10}  \\
		\hline\hline

		-   & $h$ & $B$     & $u_{LW}(h)$           & -              & -              \ignore{& Clause 11} \\
		\hline\hline
		
	\end{tabular}
	
\end{table*}


\section{Proofs for security theorems}\label{append:crypto}

\subsection{Proof for theorem \ref{the:2}}

\begin{lemma}\label{lemma:lwn} 
	If the client raises a query in the game $\Gamma^k_2$, the sequentially rational strategies of the light client $\mathcal{LW}$ (under any belief system) will not include   $L$, $R$ and $X$ (i.e., the light client will always take $T$ to report the contract the whatever it receives from the relay nodes) in this query. 
\end{lemma} 

\begin{proof} 
It is clear to see the Lemma from the recursive formulation of the utility function of $\mathcal{LW}$. 
Because no matter at any history of any information set, taking an action including $L$, $R$ or $X$ is dominated by replacing the character by $T$.
\end{proof} 

\begin{lemma} \label{lemma:2} 
	At the last query (history $h$) in the game $\Gamma^k_2$, if the client raises the last query (i.e. reaching the history $hQ$),
	when non-cooperative $\mathcal{R}_1$ and $\mathcal{R}_2$ are non-cooperative, 
	the game terminates in $hQ(attTA|a'ttTA')$,   conditioned on $d_F + p/2 >   v_i $.
\end{lemma} 

\begin{proof} 
	Let   the history $h$ denote any history where is the turn of the light client to choose from $\{Q,B\}$. 
	If the client raises the query
	due to Lemma   \ref{lemma:lwn},  it can be seen that the only reachable histories from $h$ must have the prefix of: $hQ(a|a')(t|f|x)(t|f|x)T(A|A'|O)$. 
	
	Say $hQ(a|a')(t|f|x)(t|f|x)T(A|A'|O)$ terminates the game, and w.l.o.g. let  $\mathcal{R}_i$ choose to deviate from $t$. Due to such the strategy of $\mathcal{R}_i$, a belief system of $\mathcal{R}_j$ consistent with that has to assign some non-negligible probability to the history corresponding  that $\mathcal{R}_i$ takes an action off $t$, so the best response of $\mathcal{R}_j$ after $h$ is also to take action $t$ according to the utility definitions. Then we can reason the sequential rationality backwardly: in any information set of the full nodes, the best response of the relay full nodes is to take action $t$. 
	Thus the strategy of $\mathcal{R}_i$ consistent to $\mathcal{R}_j$'s strategy must act $t$.
	Note the joint strategy $\vec{\sigma}$ that $\mathcal{R}_1$ and $\mathcal{R}_2$ always choose the action of $t$ and the $\mathcal{LW}$ always chooses the strategy of $T$. So for the light client, its belief system consistent with the fact must assign the probability of 1 to $hQatt$ out of the information set of $I^{LW}_1$ and the probability of 1 to $hQa'tt$ out of the information set of $I^{LW}_5$. Conditioned on such the belief, the light client must choose $TA$ in $I^{LW}_1$ and choose $TA'$ in $I^{LW}_5$. 
	This completes the proof for the lemma.
	
	
\end{proof}


\noindent Deferred proof for Theorem \ref{the:2}:
\begin{proof} 
	At the last query in the game $\Gamma^k_2$, the client will raise the query (namely act $Q$), conditioned on $d_L > (p+e)$.
	Additionally from Lemma \ref{lemma:2}, it immediately proves the Theorem \ref{the:2} due to backward reduction.
\end{proof}

\subsection{Proof for theorem \ref{the:1}}
\begin{lemma}\label{lemma:lwn} 
	If the client raises a query in the game $\Gamma^k_1$, the sequentially rational strategies of the light client $\mathcal{LW}$ (under any belief system) will not include   $X$ (i.e., the light client will always take $T$ to report the contract the whatever it receives from the relay nodes) in this query. 
\end{lemma} 

\begin{proof} 
	No matter at any history of any information set, taking an action including  $X$ is dominated by replacing the character by $T$,
	which is clear from the utility function.
\end{proof} 

\begin{lemma} \label{lemma:1} 
	At the last query (history $h$) in the game $\Gamma^k_1$, if the client raises the last query (i.e. reaching the history $hQ$),
	the game terminates in $hQ(atTA|a'tTA')$,   conditioned on $d_F + p-r >   v_i $, $r>v_i $.
\end{lemma} 

\begin{proof} 
	Let   the history $h$ denote the beginning of the last query. 
	If the client raises the query, the game reaches $hQ$.
	Due to Lemma   \ref{lemma:lwn},  it can be seen that the only reachable histories from $h$ must have the prefix of: $hQ(a|a')(t|f|x)T(A|A'|O)$. 
	Then it is immediate to see $t$ is dominating from the utility function.
	This completes the lemma.
	
	
\end{proof}

\noindent Deferred proof for Theorem \ref{the:2}:
\begin{proof} 
	From Lemma \ref{lemma:1}, acting $Q$ is strictly dominates   $B$ due to the utility function, at least in the last query.
	So the last query would include no deviation at all.
	It immediately allows us prove the Theorem \ref{the:1} due to backward reduction from Lemma \ref{lemma:1}.
\end{proof}

\subsection{Proof for theorem \ref{the:1-1}}

\begin{lemma}\label{lemma:lwn} 
	If the client raises a query in the augmented game $G^k_1$, the sequentially rational strategies of the light client $\mathcal{LW}$ (under any belief system) will not include   $X$ (i.e., the light client will always take $T$ to report the contract the whatever it receives from the relay nodes) in this query. 
\end{lemma} 
\begin{proof}
		No matter how the relay and the public full node act, taking an action including  $X$ is dominated by replacing the character by $T$ in the augmented game $G^k_1$.
\end{proof}

\begin{lemma} \label{lemma:1-1-1} 
	At the last query (history $h$) in the game $G^k_1$, if the client raises the last query (i.e. reaching the history $h(m|x)Q$),
	the relay node would not deviate off $t$ with non-negligible probability,   conditioned on $d_F >   v_i $.
\end{lemma}  

\begin{proof} 
	Let   the history $h$ denote any history where is the turn of the light client to choose from $\{Q,B\}$. 
	If the client raises the query
	due to Lemma   \ref{lemma:lwn},  it can be seen that the only reachable histories from $h$ must have the prefix of: $h(m|x)Q(a|a')(t|f|x)T(A|A'|O)$. 
	Given such fact, the relay would not deviate off $t$ with non-negligible probability:
	(i) deviate to play $x$ is strictly dominated;
	(ii) deviate to play $f$ with negligible probability will consistently cause the public full node acts $m$.
	This completes the lemma.
	
	
\end{proof}

\noindent Deferred proof for Theorem \ref{the:1-1}:
\begin{proof} 
	From Lemma \ref{lemma:1-1-1}, acting $Q$ is strictly dominates   $B$ due to the utility function, at least in the last query.
	Thus we can argue due to backward reduction from the last query to the first query.
	It immediately completes the proof for Theorem \ref{the:1-1}.
\end{proof}



\section{Other Pertinent Work}\label{related}
Here we thoroughly review the insufficiencies of relevant work.\footnote{Remark that there are also many studies \cite{gervais2014privacy,BITE} focus on protecting the privacy of lightweight blockchain client. However, we are emphasizing the basic functionality of the light client instead of such the advanced property.}

\smallskip
\noindent{\bf Lightweight protocols for blockchains}.
The SPV client is the first light-client protocol for Proof-of-Work (PoW) blockchains, proposed as early as Bitcoin \cite{Nak08}.
The main weakness of SPV client is that the block headers to download, verify and store increase linearly with the growth of the chain, which nowadays is $>2$ GB in Ethereum.
To realize super-light client for PoW, the ideas of Proofs of PoW   and FlyClient \cite{Flyclient} were proposed, so the light client can store only the genesis block to verify the existence of a blockchain record at a sub-linear cost, once connecting at least one honest full node.

%

Though the existence of   superlight protocols for PoW  chains, all these schemes cannot be applied for other types of consensuses such as the proof-of-stake (PoS) \cite{algorand,ouroboros,snowwhite}.
%
%
A recent study on PoS sidechains \cite{GKZ19} (and a few relevant  industry proposals \cite{Coda,Cosmos}) presented the idea of ``certificates'' to cheaply convince some always on-line clients on the incremental growth of the PoS chains; however, their obvious limitation  not only corresponds linear cost for the frequently off-line clients, but also renders serious vulnerability for the clients re-spawning from a ``deep-sleep'', because the inherent costless simulation makes the deep sleepers cannot distinguish a forged chain from the correct main chain \cite{bleeding,snowwhite}.

Few fast-bootstrap protocols  such as \cite{vault} and \cite{ouroboros-genesis} exist for PoS chains, but they are concretely designed for Algorand \cite{algorand} and Ouroboros Praos \cite{ouroboros-praos} respectively, let alone they are only suitable to boot {\em full nodes} and still incur substantial costs that are not affordable by   resource-starved clients. Different from Vault \cite{vault} and Ouroboros Genesis \cite{ouroboros-genesis}, another PoS protocol Snow White \cite{snowwhite} makes a heavyweight assumption that a permissioned list of full nodes with honest-majority can be identified, such that ``social consensus'' of these full nodes can be leveraged to allow deep sleepers to efficiently re-spawn. In contrast to \cite{snowwhite}, our design is based on a   different assumption, {\em rationality}, which is less heavy and more realistic in many real-world scenarios from our point of view.

Vitalik Buterin \cite{VitalikPoS} proposed to avoid forks in PoS  chains by incentives: a selected committee member will be punished, if she proposes two different blocks in one epoch. In this way, the light client was claimed to be supported, as the light who receives a fake block different from the correct one can send the malicious block back to the blockchain network to punish the creator of the cheating block. However, as the committees rotate periodically, it is unclear whether the protocol still works if the light client is querying about some ``ancient'' blocks generated before the relay node becomes a committee. 
%
%
A recent work \cite{GLK18} noticed the lack of incentives  due to cryptographic treatments, and proposed to use the smart contract to incorporate incentives,
but it relies on the existing lightweight-client protocols as underlying primitives, and cannot function independently. 

Another line of studies propose the concept of stateless client \cite{edrax,IOPs,Coda}. Essentially, a stateless client can validate new blocks without storing the aggregated states of the ledger (e.g., the UTXOs in Bitcoin \cite{Nak08} or the so-called states in Ethereum \cite{Woo14,buterin2014next}), such that the stateless client can incrementally update its local blockchain replica cheaply. 
Though those   studies can  reduce the burden of (always on-line) full nodes, 
they currently cannot help the frequently off-line light clients to read arbitrarily ancient records of the blockchain.


\noindent{\bf Outsourced computations}.
Reading from the blockchain can be viewed as a special computation over the ledger, which could be supported by the general techniques of outsourcing computations.
But the general techniques do not perform well.


{\em Verifiable computation} allows a  prover to convince a verifier that an output is obtained through correctly computing a function \cite{GGP10}.
%
Recent constructive development of generic verifiable computations such as SNARKs \cite{BCG13} enables   efficient verifications (for general NP-statements) with   considerable cost of proving. Usually, for some heavy statements such as proving the PoWs/PoSs of a blockchain, the generic verifiable computation tools are infeasible in practice \cite{Flyclient,Coda}.

To reduce the high cost of generality, Coda \cite{Coda} proposed to use the idea of incremental proofs to allow (an always on-line) stateless client to validate blocks cheaply, but it is still unclear how to use the same idea to convince the (frequently off-line) clients without prohibitive proving cost. Some other studies \cite{versum,al2018fraud} focus on allowing the resource-constrained clients to efficiently verifying the computations taking the blockchain as input, but they require the clients already have the valid chain headers, and did not tackle the major problem of light clients, i.e., how to get the correct chain of headers. 

{\em Attestation via trusted hardwares} has attracted many attentions recently \cite{CD16,ZCC16}, and it becomes enticing to employ the trusted hardwares towards the practical light-client scheme \cite{BITE}.
However, recent Foreshadow attacks \cite{foreshadow} put SGX's attestation keys in danger of leakage,
and potentially allow an adversary to forge attestations, which might fully break the remote attestation of SGX and subsequently challenge the fundamental assumption of trusted hardwares.

{\em Outsourced computation via incentive games} were discussed before \cite{PKC14,Kupcu17}.
Some recent studies even consider the blockchain to facilitate games for outsourced computations \cite{DWA17,truebit}. All the studies assume an implicit game mediator (e.g. the blockchain) who can speak/listen to all involving parties, including the requester, the workers and a trusted third-party (TTP). However, in our setting, we have to resolve a special issue that there is neither TTP nor mediator, since the blockchain, as a potential candidate, can neither speak to the light nor validate some special blockchain readings.

\end{document}